\documentclass[twoside]{article}

\usepackage[preprint]{aistats2027}
\runningtitle{Differentiable Repairability and Mutations}
\runningauthor{Sun et al.}
\usepackage[utf8]{inputenc} 
\usepackage[T1]{fontenc}    
\usepackage{hyperref}       
\usepackage{url}            
\usepackage{booktabs}       
\usepackage{amsfonts}       
\usepackage{nicefrac}       
\usepackage{microtype}      
\usepackage{xcolor}         
\usepackage{amsmath}
\usepackage{amssymb}
\usepackage{xcolor}
\usepackage{mathtools}
\usepackage{amsthm}
\usepackage{dsfont}
\usepackage{enumitem}
\usepackage{algorithm}
\usepackage{subcaption}
\theoremstyle{plain}
\newtheorem{theorem}{Theorem}
\newtheorem{proposition}{Proposition}

\theoremstyle{definition}
\newtheorem{definition}{Definition}

\theoremstyle{remark}
\newtheorem{remark}{Remark}
\begin{document}

%

%

\twocolumn[

\aistatstitle{Differentiable Repairability \& Mutations for Scalable Combinatorial Optimization}

\aistatsauthor{
Yongliang Sun\\
\mdseries Michigan State University
\And
Ismail R. Alkhouri\\
\mdseries LANL \& University of Michigan
\And
Cheng-Han Huang\\
\mdseries Michigan State University
\AND
Alvaro Velasquez\\
\mdseries University of Colorado
\And
Susmit Jha\\
\mdseries SRI International
\And
Rongrong Wang\\
\mdseries Michigan State University
}
\vspace{0.99cm}
]

\begin{abstract}
Combinatorial optimization problems such as Maximum Independent Set (MIS) and Maximum Cut (MaxCut) arise across a wide range of applications, yet solving large-scale instances remains computationally challenging. While discrete local-search procedures can produce high-quality solutions, their inherently sequential and iterative nature can become increasingly expensive as problem size and density grow. In this paper, we present a differentiable perspective on combinatorial optimization that leverages continuous relaxation with parallelizable gradient-based updates to improve scalability while reducing dependence on discrete local search. More specifically, for MIS, we show that solutions that remain repairable through a discrete (1,2)-swap can nevertheless be fixed points of projected gradient ascent, revealing a fundamental limitation of purely local gradient dynamics. For MaxCut, we develop a new quadratic formulation whose landscape provably allows projected gradient ascent to escape 1-flip-repairable solutions while ensuring binary local maximizers. Building on these insights, \textcolor{black}{and inspired by discrete perturb-and-improve methods, we introduce a framework that combines projected gradient ascent with \textit{differentiable} mutations (partial resets) to \textit{efficiently} explore new regions of the solution space}. Experimental results on MIS and MaxCut demonstrate the potential of this approach for scalable combinatorial optimization, particularly in regimes where conventional discrete local search becomes computationally prohibitive.
\end{abstract}


\section{Introduction}

Combinatorial optimization problems (COPs) arise in a wide range of applications, encompassing tasks from scheduling to graph partitioning and beyond \cite{glover2019quantum}. Among the most prominent COPs are the Maximum Independent Set (MIS) and the Maximum Cut (MaxCut) problems \cite{karp1972reducibility}, which appear in diverse domains such as frequency assignment and interference management in wireless networks \cite{Matsui2000FrequencyAssignment, gu2024graph}, task scheduling \cite{eddy2021maximum}, genome sequencing \cite{joseph1992dna, Zweig2006OnTM}, physical design of VLSI circuits \cite{liers2011via}, and phylogenetic tree construction in computational biology \cite{4015375}.

A standard approach to tackling these problems is to formulate them as Integer Linear Programs (ILPs) \cite{nemhauser1975vertex} or as Quadratic Unconstrained Binary Optimization (QUBO) problems \cite{glover2019quantum}. However, since both are integer programs, scalability remains a major challenge, and the run-time required to obtain optimal solutions can grow prohibitively large. As a result, problem-specific heuristic methods, which do not provide optimality guarantees, have been developed. \textcolor{black}{Examples include the Reduction-based method for MIS (ReduMIS) \cite{lamm2016finding} and Breakout Local Search (BLS) for MaxCut \cite{benlic2013breakout}, \textcolor{black}{both of which rely heavily on local search procedures (such as \cite{andrade2012fast}). Local search methods initialize with a candidate solution and iteratively seek to improve it. While they can produce high-quality solutions, their inherently sequential and iterative nature can make finding such solutions increasingly computationally expensive as problem size and density grow.}}

More recent studies have explored using Machine Learning (ML) techniques to solve COPs (see the ML4CO and Neural CO (NCO) survey papers \cite{ANGIONI2025107102,pmlr-v176-gasse22a}). While many of these methods obtain solutions significantly fast (relative to heuristics) at inference, most of them suffer from generalization \cite{alkhouri2025differentiable} (unlike very few such as the ANYCSP method in \cite{tonshoff2023one}). In most cases, heuristics and IP commercial solvers (especially with sparse problem instances) achieve the state-of-the-art (SOTA) results to a degree that supervised learning methods use heuristics to label their training data \cite{li2018combinatorial}. 


Applying gradient-based optimizers to box-constrained relaxed QUBOs does not suffer from the same scalability limitations as integer formulations \textcolor{black}{and local search} \cite{alkhouri2025differentiable, alkhouri2026MaxCut}. However, the relaxed QUBO landscape is highly non-convex, and gradient-based methods typically stall at local optima. This significantly limits their ability to explore the solution space \textit{efficiently}. With the recent advances in parallel computing, particularly GPUs, two main strategies have been proposed to improve exploration: (i) GPU-based sampling methods \cite{sun2023revisiting, resco} and (ii) parallelized initializations methods \cite{alkhouri2025differentiable, alkhouri2026MaxCut}. Nevertheless, reliance on first-order optimization in such non-convex landscapes, and its tendency to get trapped in local optima, remains a key bottleneck. In general, while these approaches may outperform SOTA heuristics in certain regimes, they can also underperform in others, as recently demonstrated for MIS in \cite{wu2025unrealized}.

\textcolor{black}{Motivated by these observations, we provide a detailed analysis of the relaxed non-convex QUBO landscapes for both the MIS and MaxCut problems from repairability perspective, and introduce new quadratic objectives for MaxCut. Furthermore, we propose a framework that enables more effective exploration through a mechanism that combines an iterative, differentiable, and mutation-based global reset such that adopting discrete local search procedures become less necessary. }

\textcolor{black}{From a different angle, our motivation of studying (and developing) continuous relaxation of QUBO is that discrete search is inherently sequential and is typically executed on CPUs, since each accepted node changes the gains of neighboring variables and requires repeated updates. In contrast, the continuous relaxation is naturally expressed through matrix-vector operations and can be efficiently parallelized on GPUs. Our goal is precisely to use this parallelizable continuous optimization to replace a substantial portion of the computational load that would otherwise be handled by iterative discrete local search. We summarize our contributions as follows:}

\begin{itemize}[wide=0pt]

    \item \textbf{Theory of Differentiable Repairability.} We provide theoretical analysis on the relaxed non-convex QUBOs for MIS \textcolor{black}{and MaxCut by first introducing the notion of \emph{locally repairable} local maxima, and showing that gradient-based optimizers can stall at such points for the standard relaxed QUBO of MIS and MaxCut}. Motivated by this observation, we propose a new quadratic formulation for MaxCut and theoretically (\textcolor{black}{and empirically}) show that it provides gradient-based methods the ability to escape one-node locally repairable local maxima. 
    
    \item \textbf{Differentiable Mutations Algorithm.} Guided by our analysis, we propose an algorithm that efficiently improves exploration by a differentiable mutation-driven global reset \textcolor{black}{(with optional discrete local search)}. \textcolor{black}{As such, we term the algorithm mutation-guided Combinatorial Quadratic Optimization (mQO). Given its differentiable nature, we develop CPU-only and GPU+CPU versions of mQO. When compared to SOTA heuristics, mQO's dependence on discrete local search is significantly reduced.} 
    
    \item \textbf{Extensive Evaluation.} We conduct extensive experiments on small- and large-scale Erdős--Rényi (ER) graphs (with varying densities), and other benchmark graphs, ranging from 100 to 40,000 nodes and 2,500 to 639,984,000 edges (\textcolor{black}{which are significantly larger than most benchmarks in recent studies \cite{resco}}). We demonstrate that our algorithms allow gradient-based optimizers to either outperform or provide faster solutions against SOTA heuristics, ILP solvers, and other baselines, especially on large-scale graphs. 

\end{itemize}

\subsection{Preliminaries}\label{sec: prelim}


\paragraph{Notations:} An undirected graph is represented as $G=(V,E)$, where $V$ is the set of nodes and $E\subseteq V \times V$ is the set of edges. 
The number of nodes (resp. edges) is $|V| = n$ (resp. $|E| = m$), where $|\cdot|$ denotes the cardinality. 
Unless otherwise stated, we use $\mathcal{N}(v) = \{u\in V \mid (u,v)\in E\}$ to represent the set of neighbors of node $v$. 
The degree of a node $v\in V$ is $\textrm{d}(v) = |\mathcal{N}(v)|$, and the maximum degree of the graph is $\Delta$. 
The adjacency matrix of graph $G$ is denoted by $\mathbf{A}\in \{0,1\}^{n \times n}$. The degree diagonal matrix is $\mathbf{D}$. We use $\mathbf{I}$ to denote the identity matrix. 
The trace of a matrix $\mathbf{A}$ is denoted by $\mathrm{tr}(\mathbf{A})$. 
For any positive integer $n$, $[n]:=\{1, \ldots, n\}$. The vector of all ones and size $n$ is denoted by $\mathbf{1}_n$. The $i^{\textrm{th}}$ canonical basis vector is denoted by $\mathbf{e}_i$, which is the vector of all zeros except at the $i^{\textrm{th}}$ index, where the value is one. Furthermore, we use $\mathds{1}(\cdot)$ to denote the indicator function that returns $1$ (resp. $0$) when its argument is True (resp. False).

In this paper, we consider the MIS and MaxCut problems, which we formally define next. 

\begin{definition}[MIS]
  \label{def: MIS}
  Given an undirected graph $G = (V, E)$, the goal of MIS is to find a subset of vertices $I \subseteq V$ such that there are no edges connecting any nodes in $I$, and $|I|$ is maximized.
\end{definition}

\textcolor{black}{A maximal independent set (MaxIS) cannot be enlarged by adding another node, whereas a maximum independent set (MIS) has the largest cardinality among all independent sets (ISs). Thus, every MIS is maximal, but the converse does not necessarily hold.} 

Let $\mathbf{z}_v$ be an entry of $\mathbf{z}\in \{0,1\}^n$ that corresponds to a node $v\in V$. The MIS ILP is \cite{nemhauser1975vertex}: $\max_{\mathbf{z}\in \{0,1\}^n}  \sum_{v \in V} \mathbf{z}_v~~ 
\text{s.t.}~~ \mathbf{z}_v + \mathbf{z}_u \leq 1\:, \forall (v,u) \in E.$
\begin{definition}[MaxCut] 
Given an undirected graph $G=(V,E)$, the goal of MaxCut is to partition $V$ into two non-overlapping sets $S$ and $\bar{S} = V \setminus S$ such that the number of edges crossing the cut (i.e., with one endpoint in $S$ and the other in $\bar{S}$) is maximized.
\end{definition}

Let $\mathbf{y}_{v,u}$ corresponds to edge $(v,u)$, stacked in an $m$-dimensional vector. The MaxCut ILP is \cite{williamson2011design}: $\max_{\mathbf{z}\in\{0,1\}^{n},~\mathbf{y}\in\{0,1\}^{m}}
~ \sum_{(v,u)\in E} \mathbf{y}_{v,u}
~\text{s.t.}~
\mathbf{y}_{v,u} \leq \mathbf{z}_{v} + \mathbf{z}_{u},
~\mathbf{y}_{v,u} \leq 2 - \mathbf{z}_{v} - \mathbf{z}_{u},
~ \forall (v,u) \in E \:.$

Most standard MIS and MaxCut heuristic methods, such as \cite{lamm2016finding} for MIS and \cite{benlic2013breakout} for MaxCut, typically use discrete local search procedures that can improve the solution. Next, we formally define the $(1,2)$-swap procedure for MIS \cite{andrade2012fast}, the $1$-flip \cite{benlic2013breakout}, and $2$-flip \cite{two_flip_max_cut_local_search} searches for MaxCut, through the binary points that can improve if these steps are applied. We term such points as ``\textbf{repairable}'' points. 
%


\begin{definition}[MIS $(1,2)$-swap repairable point]
    \label{def: 1,2 swap}
    Let $\mathbf{x}\in \{0,1\}^n$ be a binary vector such that $I = \{v\in V : \mathbf{x}_v = 1\}$ is a maximal IS in graph $G=(V,E)$. Vector $\mathbf{x}$ is $(1,2)$-swap repairable if there exists some $\mathbf{x}_v=1$ such that if we remove $v$, and add two of its non-adjacent neighbors (to each other), that is $u,w \notin I$ and $u,w \in \mathcal{N}(v)$, the new set $I' = (I\setminus \{v\}) \cup \{u,w\}$ is independent and its cardinality increased by one, i.e., $|I'| = |I|+1$. 
\end{definition}

%
%
\begin{definition}[MaxCut $1$-flip repairable point]
    \label{def: 1-flip}
    Let $\mathbf{x}\in \{-1,1\}^n$ be a binary vector such that $S = \{v\in V : \mathbf{x}_v = 1\}$ corresponds to a cut in graph $G=(V,E)$ of value $\texttt{Cut}(S) := \sum_{v\in S} \sum_{u\in V \setminus S} \mathds{1}((v,u)\in E)$. Vector $\mathbf{x}$ is $1$-flip repairable if there exists some $\mathbf{x}_v=1$ such that if we place $v$ in set $\bar{S}$ instead of $S$ (i.e., we have $S' = S\setminus \{v\}$), the cut value will increase, that is, $\texttt{Cut}(S')>\texttt{Cut}(S)$.
\end{definition}
%
%

\begin{definition}[MaxCut $2$-flip repairable point]
    \label{def: 2-flip}
Let $\mathbf{x}$, $S$, and $G$ be as given in Definition~\ref{def: 1-flip}. Vector $\mathbf{x}$ is $2$-flip repairable if there exists some $u \in S$, $v \notin S$ and $(u,v) \in E$ with $\mathbf{x}_v \neq \mathbf{x}_u$ such that if we place $u$ in set $\bar{S}$ and $v$ in set $S$, i.e., we have $S' = (S\setminus \{u\}) \cup \{v\}$ and $\bar{S}' = (\bar{S}\setminus \{v\}) \cup \{u\}$, the cut value will increase, that is, $\texttt{Cut}(S')>\texttt{Cut}(S)$.
\end{definition}
%


\section{Motivation} \label{sec: motivation}

\textcolor{black}{In this section, we describe the computational complexity of discrete local search, and motivate our subsequent analysis and algorithmic development based on continuous relaxations.}

\paragraph{Per-pass complexity of local search:}
For \textbf{MIS $(1,2)$-swap,} if $I$ is a MaxIS, then for each $v \in I$, in the worst case, this requires enumerating $O(\textrm{d}^2(v))$ pairs per node, leading to a total complexity of $O\!\left(\sum_{v\in I} \textrm{d}^2(v)\right) \le O(|I|\Delta^2).$
For \textbf{MaxCut 1-flip,} for each $v \in V$, the gain from flipping can be computed by inspecting its neighbors, which costs $O(\textrm{d}(v))$.
Summing over all nodes yields a total complexity of $O\!\left(\sum_{v\in V} \textrm{d}(v)\right) = O(m).$
A naive implementation of \textbf{MaxCut 2-flip} considers pairs of nodes (typically edges), and evaluates the gain of swapping them.
This leads to a complexity of $O(m\Delta)$, since each evaluation involves scanning neighborhoods. With incremental gain updates, this can be reduced to $O(m)$ amortized per full pass.

\paragraph{Cost of repeated local improvement:}
While a single pass of local search (e.g., the MaxCut $1$-flip procedure) requires $O(m)$ time, the main computational cost arises from the need to apply the procedure iteratively. Indeed, after performing one improving flip, new nodes may become flippable, and thus multiple passes over the nodes are required until no further improvement is possible. In the worst case, the number of such iterations depends on the number of successful flips before convergence. Therefore, the overall complexity of the iterative $1$-flip local search can be written as $O(T\cdot m)$, where $T$ denotes the number of improvement rounds. This observation suggests that, despite their simplicity, local search procedures can become computationally expensive, especially for large-scale graphs.

\paragraph{From local search to continuous optimization:}
\textcolor{black}{To mitigate this issue, based on the local maxima characterization developed in the next section, in Section~\ref{sec: proposed algorithm}, we propose a continuous relaxation framework and apply projected gradient ascent (PGA), enabling more efficient exploration of the solution space while reducing reliance on repeated discrete local updates. Moreover, the continuous formulation naturally supports parallelization on modern hardware (e.g., GPUs), allowing scalable optimization through batched operations. To our knowledge, this is the first work to theoretically analyze relaxed QUBO objectives from the perspective of local search (i.e., repairable points) and use that analysis to design an objective specifically for reducing dependence on local search. The contribution is therefore not the use of a standard QUBO or the addition of a standard local-search routine, but rather the theory-guided design of the continuous relaxation so that highly parallelizable gradient optimization can take over work that would otherwise require sequential discrete search.}

\section{Gradient-Based Repairability}\label{sec: theory}




\subsection{Overview}

To understand the computational advantages of PGA, it is essential to characterize which types of local search procedures can be \emph{effectively avoided} through continuous relaxation.
In particular, we are interested in identifying classes of \emph{repairable points} that PGA can escape without requiring explicit discrete updates.


\textcolor{black}{Towards that goal, in this section,} we first characterize local maxima of the MIS QUBO with respect to (w.r.t.) repairable points. We then turn to MaxCut, where we provide a detailed analysis and introduce a new quadratic formulation. All proofs are deferred to Appendix~\ref{appen: proofs}.

For MIS, a relaxed standard QUBO is given by~\cite{pardalos1992branch}:
\begin{equation}
\label{eqn: MIS QUBO}
\begin{aligned}
\max_{\mathbf{x}\in [0,1]^n} h(\mathbf{x})
:= \sum_{v\in [n]} \mathbf{x}_v - \gamma \sum_{(v,u)\in E} \mathbf{x}_u \mathbf{x}_v\:,
\end{aligned}
\end{equation}
where $\gamma>1$ is the edge-penalty parameter~\cite{mahdavi2013characterization}. In matrix form, we have $h(\mathbf{x}) = \mathbf{1}^T \mathbf{x} - \frac{\gamma}{2} \mathbf{x}^T \mathbf{A} \mathbf{x}\:.$ For MaxCut, the \textcolor{black}{relaxed} standard QUBO is~\cite{williamson2011design}:
\begin{equation}
\label{eqn: Cut QUBO}
\max_{\mathbf{x}\in [-1,1]^n} f_\textrm{L}(\mathbf{x}) := \frac{1}{2} \sum_{(v,u)\in E} (1-\mathbf{x}_u \mathbf{x}_v). \:,
\end{equation}
In the discrete domain, $f_\textrm{L}(\mathbf{x})= \frac{1}{4}\mathbf{x}^T \mathbf{L} \mathbf{x}$, where $\mathbf{L} = \mathbf{D} - \mathbf{A}$ is the graph Laplacian. We refer to $f_{\textrm{L}}$ as the \textbf{Laplacian} formulation.

Applying PGA to these relaxations (as well as to the formulations introduced later) yields the generic update: $\mathbf{x} \leftarrow \Pi(\mathbf{x} + \alpha \mathcal{G})$,
where $\mathcal{G}$ denotes a gradient step (w.r.t. the objective), $\alpha>0$ is the step size, and $\Pi$ is the projection onto the corresponding box constraint (i.e., $[0,1]^n$ or $[-1,1]^n$). In this paper, we use ``PGA'' to broadly refer to gradient-based optimizers, including momentum-based gradient ascent~\cite{polyak1964some} and Adam~\cite{kingma2015adam}.

Running PGA on $h$ or $f_{\textrm{L}}$ may converge to boundary (and hence binary) points corresponding to feasible MIS or MaxCut solutions. More generally, PGA fixed points can be classified as: (i) \textbf{non-stationary boundary points}, where $\mathcal{G} \neq \mathbf{0}$ but projection prevents further updates, and (ii) \textbf{stationary points}, where the gradient vanishes (i.e., $\nabla h = \mathbf{0}$ or $\nabla f_{\textrm{L}} = \mathbf{0}$).


\textcolor{black}{
\begin{remark}
    In this paper, we primarily focus on characterizing local maximizers and do not consider strict saddle points that are not observed in practice (such as the all-zero vector), largely due to the use of momentum-based optimizers that help avoid convergence to such points.
\end{remark}}





\subsection{MIS Local Repairability}

\textcolor{black}{In \cite{alkhouri2025differentiable}, it was shown that all PGA fixed boundary points correspond to MaxISs (Lemma~2) using a clique-based version of $h$. It was also shown that PGA can improve on non-maximal ISs (which can be viewed as $1$-flip improvable points), but may converge to local maxima corresponding to suboptimal MaxISs (proof of Theorem~2). These results straightforwardly apply to $h$.  }


\textcolor{black}{Here, we are interested in the question: \textit{Can PGA escape $(1,2)$-swap points of $h$ in \eqref{eqn: MIS QUBO}?} We answer this question in the negative:} PGA can stall even at $(1,2)$-swap repairable points. 
\begin{proposition}
\label{prop: PGA is stuck at (1,2)}
Let $\mathbf{x}$ be a MIS $(1,2)$-swap repairable point. Then, $\mathbf{x} = \Pi_{[0,1]^n}(\mathbf{x} + \alpha \nabla_{\mathbf{x}} h)$.
\end{proposition}

\textcolor{black}{For MIS, this result motivates the use of other exploration mechanisms, such as the partial reset (differentiable mutations) discussed in Section~\ref{sec: proposed algorithm}}




\subsection{MaxCut Local Repairability \& a New Formulation}


We now turn to MaxCut and start with the standard Laplacian formulation $f_{\textrm{L}}$ from \eqref{eqn: Cut QUBO}. \textcolor{black}{Setting the gradient to 0 for $f_\textrm{L}$ in \eqref{eqn: Cut QUBO} mediately shows that any point in the null space of $\mathbf{L}$ (which is any point in $\{c\mathbf{1}_n, ~c\in [-1,1]\}$) is a stationary point and hence PGA fixed point. As $c\in [-1,1]$, these are boundary and interior points \cite{alkhouri2025scalable}.} 

\textcolor{black}{The question is: \textit{Can PGA escape $1$-flip repairable points of $f_\textrm{L}$?} We also answer in the negative using the following Theorem.}
\begin{theorem}
\label{th: all fixed of perturbed laplacian}
Let $\mathbf{x}\in\{-1,1\}^n$ be a $1$-flip repairable point. Then, $\mathbf{x} = \Pi_{[-1,1]^n}(\mathbf{x} + \alpha \nabla_{\mathbf{x}} f_\textrm{L})$.
\end{theorem}
\textcolor{black}{The question then becomes: \textit{Can we design an objective such that PGA can escape $1$-flip repairable points?}} To address this, we consider an alternative formulation based on the \textbf{adjacency} matrix:
\begin{equation}
\label{eqn: adj cut formualtion}
f_\textrm{A}(\mathbf{x}) := -\mathbf{x}^T \mathbf{A} \mathbf{x}\:.
\end{equation}
%
Over the discrete domain $\{-1,1\}^n$, $f_{\mathrm{A}}$ and $f_{\mathrm{L}}$ differ only by a constant, but their continuous relaxations behave differently. \textcolor{black}{In the following result, we show that PGA with $f_\textrm{A}$ can escape $1$-flip repairable points.}
%
\begin{proposition}
\label{prop: adj form is not 1-flip repair}
If $\mathbf{x}\in\{-1,1\}^n$ is a $1$-flip repairable point, then it is not a PGA-fixed point w.r.t. $f_\textrm{A}$.
\end{proposition}
Thus, the adjacency formulation improves exploration by allowing escape from at least $1$-flip repairable points. However, it introduces non-binary points that we describe in the next proposition (and its proof in Appendix~\ref{appen: proofs}) and subsequent example.
\begin{proposition}
\label{prop: adj_fractional_localmin}
\textcolor{black}{For $f_\textrm{A}(\mathbf{x})$, there could exist non-zero and non-binary local maximizers in $[-1,1]^n$}.
\end{proposition}


\textcolor{black}{The points described in Proposition~\ref{prop: adj_fractional_localmin} can be constructed for different graphs. For example, take $G=(\{1,2,3,4\},\{(1,2),(1,3),(2,3),(2,4)\})$ (a triangle with an extra edge), $\mathbf{x}=(t,1,-1,-1)$ with $t\in(-1,1)$, and $\mathbf{x}'=(\pm1,1,-1,-1)$.
Then $f_{\textrm{A}}(\mathbf{x})=f_{\textrm{A}}(\mathbf{x}')=4$, the global maximum over $[-1,1]^4$.}


Driven by the aforementioned characterization of the local maxima of the two MaxCut objective functions, and towards maintaining (i) \textcolor{black}{the elimination of unwanted points that might be introduced in $f_\textrm{A}$}, and (ii) the PGA ability to escape from $1$-flip repairable points, we propose a new objective that uses the adjacency matrix with a linear perturbation-based \textbf{bias} term: 
\begin{equation}
\label{eqn: adj_lin_two_props}
f_\textrm{B}(\mathbf{x}) := -\lambda \mathbf{1}^T\mathbf{x} - \mathbf{x}^T\mathbf{A}\mathbf{x}\:.
\end{equation}
%
%
%
\begin{theorem}
\label{thm: adj_lin_two_props}
For $f_\textrm{B}(\mathbf{x})$, set $\lambda \in (0,2)$, then:
(i) PGA does not stall at $1$-flip repairable points, and 
(ii) all local maximizers over $[-1,1]^n$ are binary. 
\end{theorem}
%

\textcolor{black}{In Appendix~\ref{sec: append cut obj comparison}, we empirically show the overall improvements achieved by using $f_{\mathrm B}$ compared to the other two MaxCut objectives.}

\begin{remark}
\textcolor{black}{In Appendix~\ref{sec: append empirical support of aflip escap of fB}, we show the impact of adopting discrete local search in our algorithm for $f_\textrm{B}$ vs. $f_\textrm{L}$. We observe lower number of local search improvements when we use our $f_\textrm{B}$ vs. $f_\textrm{L}$ which indicates better differentiable exploration. This is because PGA with $f_\textrm{B}$ has the ability to escape $1$-flip repairable points, unlike $f_\textrm{L}$ where $1$-flip repairable points can only improve by discrete local search. This empirically supports the $1$-flip escapability property from Theorem~\ref{thm: adj_lin_two_props}.}  
\end{remark}

As a summary, although PGA applied to $f_\textrm{B}$ has the ability to escape $1$-flip repairable points, it may still stall at $2$-flip repairable points (from Definition~\ref{def: 2-flip}). To further enhance exploration, \textcolor{black}{inspired by discrete perturbations in Iterated Local Search (ILS) methods (such as \cite{lourencco2003iterated}}), in the next section, we propose using a \textit{differentiable} reset (mutations) mechanism that has the potential to efficiently escape such local maxima.


\section{Differentiable Mutations}\label{sec: proposed algorithm}


In this section, we present our MIS and MaxCut algorithms, motivated by the need to increase the search radius, particularly when the corrections offered by discrete local search become computationally expensive or fail to provide improvements. 
\begin{algorithm*}[t]
\caption{\textbf{mQO-MIS}.}
\textbf{Input}: Graph $G$, initial $\mathbf{x}$, gradient operator $\mathcal{G}$ w.r.t. the MIS function $h$ in \eqref{eqn: MIS QUBO}, reset parameter $\rho$, step size $\alpha$, number of global search rounds $T_{\mathrm{gs}}$, and $I_{\textrm{sol}}=\{\cdot\}$ (current best set) \\
\vspace{1mm}
\textbf{Output}: The best obtained solution set $I_{\textrm{sol}}$\\
\vspace{1mm}
\small{01:} \textbf{While} time budget permits \\
\vspace{1mm}
\small{02:} \hspace{5mm}\textbf{Initialize} $\mathbf{x}$ as
$\mathbf{x} \leftarrow \Pi_{[0,1]^n}(\mathbf{d}_{\text{base}} + \boldsymbol{\epsilon})$,
where $\mathbf{d}_{\text{base},v} = 1 - \frac{\textrm{d}(v)}{\Delta}$ and
$\boldsymbol{\epsilon} \sim \mathcal{N}(0, \sigma^2 \mathbf{I})$ \\
\vspace{1mm}
\small{03:} \hspace{5mm}\textbf{Obtain} $\mathbf{x}\leftarrow \Pi_{[0,1]^n}(\mathbf{x}+\alpha \mathcal{G})$ until $I = \{v\in V : \mathbf{x}_v > 0.5\}$ is MIS, detected via \eqref{eqn: MIS checking} with binarized $\mathbf{x}$ \\
\vspace{1mm}
\small{04:} \hspace{5mm}\textbf{Obtain} $I = \{v\in V : \mathbf{x}_v = 1\}$ and \textbf{update} best set as $I_{\textrm{sol}} \leftarrow I$ only \textbf{If} $|I_{\textrm{sol}}| < |I|$ \\
\vspace{1mm}
\small{05:} \hspace{5mm}\textbf{For} $T_{\mathrm{gs}}$ rounds (global reset loop) \\
\vspace{1mm}
\small{06:} \hspace{12mm}\textbf{Obtain} $\mathbf{x}$ from $I_{\textrm{sol}}$ such that each coordinate is obtained as $\mathbf{x}_v = \mathds{1}(v\in I_{\textrm{sol}})$\\
\vspace{1mm}
\small{07:} \hspace{12mm}\textbf{Obtain} set $R\subset V$ by randomly choosing $\lfloor\rho n\rfloor$ indices\\
\vspace{1mm}
\small{08:} \hspace{12mm}\textbf{Set} $\mathbf{x}_v \leftarrow 0,\ \forall v\in R$ (reset for global search)\\
\vspace{1mm}
\small{09:} \hspace{12mm}\textbf{Obtain} $\mathbf{x}\leftarrow \Pi_{[0,1]^n}(\mathbf{x}+\alpha \mathcal{G})$ until a MIS is detected (optimizer gradient updates)\\
\vspace{1mm}
\small{10:} \hspace{12mm}\textbf{Obtain} $I' = \{v\in V : \mathbf{x}_v > 0\}$ and \textbf{update} best set as $I_{\textrm{sol}} \leftarrow I'$ only \textbf{If} $|I_{\textrm{sol}}| < |I'|$\\
\vspace{1mm}
\small{11:} \hspace{5mm}\textbf{Obtain} $I \leftarrow \texttt{OneTwoSwap}(I_{\textrm{sol}},G)$ and \textbf{update} $I_{\textrm{sol}} \leftarrow I$ only \textbf{If} $|I_{\textrm{sol}}| < |I|$ (local search)\\
\vspace{1mm}

\vspace{-3.5mm}
\label{alg: main mQO-MIS}
\end{algorithm*}


After the optimizer becomes stuck at a binary local maximizer $\mathbf{x}$, we randomly select a subset of nodes $R\subset V$ using a uniform distribution where
\[
R\sim \text{Unif}\{U \subset V :  |U| = \lfloor \rho n \rfloor\}.
\]
Here, $\rho \in (0,1)$ is the \textbf{reset parameter} that determines the percentage of nodes that will be reset. Then, we set all entries of $\mathbf{x}$ that belong to $R$ to $0$. In other words, we set $\mathbf{x}_v = 0, \forall v\in R$. 

Setting a subset of indices of $\mathbf{x}$ to $0$, for the MIS problem, corresponds to converting the maximal IS to a non-maximal IS as the nodes with value $0$ still do not violate the independence property. For MaxCut, setting a subset of indices of $\mathbf{x}$ to $0$ corresponds to \textit{implicitly} setting these nodes to undecided states (neither belonging to $S$ nor to $\bar{S}$). Thus, $\mathbf{x}_v, \forall v\in R$, can evolve towards $+1$ or $-1$, i.e., the boundary of the box constraint, depending on the new (reset) initial point and the gradient direction.

After setting $\mathbf{x}_v = 0, \forall v\in R$, the gradient-based optimizer is run again from the updated initialization until converging to a (possibly different) local maximizer $\tilde{\mathbf{x}}$. The reset is \emph{accepted} only if it improves the discrete objective, i.e., 
\[
    h(\tilde{\mathbf{x}}) > h(\mathbf{x})~~\text{for MIS},~~~\text{or}~~~ f_\textrm{L}(\tilde{\mathbf{x}}) > f_\textrm{L}(\mathbf{x})~~\text{for MaxCut},
\]
otherwise, a new random reset set is drawn from the current solution and the process repeats until the time budget expires. We note that the cut value is evaluated using $f_\textrm{L}$, while optimization is performed using $f_\textrm{B}$, since $f_\textrm{B}$ no longer directly corresponds to the cut value. 

\textcolor{black}{\begin{remark}
    The aforementioned unbiased (and partial) global reset strategy (i.e., the case of $\rho\in (0,1)$ with resetting the indices of the variables in $R$ to $0$) is much more effective than a full random restart (i.e., the case of $\rho = 1$ with re-initializing $\mathbf{x}$. We support this empirically in Appendix~\ref{sec: append partial reset > random re-initialization}
\end{remark}}

For MIS, we use a degree-based initialization to initialize vector $\mathbf{x}$ (given as $\mathbf{d}_{\textrm{base}}$ in Step~2) in Algorithm~\ref{alg: main mQO-MIS} based on the intuition that nodes with higher degree are less likely to be in the independent set \cite{alkhouri2022differentiable}. 

Step~3 performs the gradient updates of the optimizer w.r.t. \eqref{eqn: MIS QUBO} until a MIS is detected, for which, we check whether the set $I$ extracted from $\mathbf{x}$ corresponds to a MaxIS using the gradient-based MIS verification method proposed in \cite{alkhouri2025differentiable} (by checking if the binary converged point is a PGA-fixed point), that avoids iterating over all nodes and relies only on matrix–vector multiplications. More specifically, given a binary $\mathbf{x}$, we check whether the following returns True or False
\begin{equation}\label{eqn: MIS checking}
 \mathds1\big\{ \mathbf{x} = \Pi_{[0,1]^n}(\mathbf{x} + \alpha \nabla h(\mathbf{x}) ) \big\}\:.   
\end{equation}
In Step~4, we update the best set if a better MIS is found. Steps~5 to 10 apply the global reset for $T_{\textrm{gs}}>1$ rounds. More specifically, in Step 7, we randomly select the set $R$, then in Step~9 and Step~10, we run the gradient updates again and update the best set if the converged one is of larger size. In Step~11, the iterative $(1,2)$-swap local search is applied. If the resulting set has higher cardinality, then the best set is updated. 



For MaxCut, the initialization of $\mathbf{x}$ in Algorithm~\ref{alg: main mQO-Cut} (which is deferred to Appendix~\ref{appen: implementation of pmQO}) is also a degree-based initialization strategy since high-degree nodes have greater influence on the cut value when switching between two partitions (as described in Step~2). The remaining procedure of MaxCut is similar to Algorithm~\ref{alg: main mQO-MIS}, with few changes: (i) We use $f_\textrm{B}$; (ii) A checker is not needed as any binary vector (that is not $\mathbf{1}_n$ or $-\mathbf{1}_n$) returns a non-zero cut value, therefore, the gradient updates are run until convergence; (iii) We use a combination of one $1$- and $2$-flip local search procedures.


\textcolor{black}{In Appendix~\ref{sec: append detailed local search}, we describe how we apply the local search procedures on MIS (resp. Maxcut) $I$ (resp. $S$), denoted in our algorithms as $\texttt{OneTwoSwap}(I,G)$ (resp. $\texttt{OneTwoFlip}(S,G)$).}



\begin{remark}
    \textcolor{black}{Our algorithms overall combine three components: (i) gradient updates, (ii) global reset (determined by $\rho$), and (iii) adopted discrete local search. However, performance gains come primarily from the differentiable updates and then the global resets, while the contribution of discrete local search is relatively very small. This is supported by the gain breakdown in the results of Table~\ref{tab:progress_breakdown} where we show that discrete local search only contributes to less than $2\%$ of the final MIS size. We further compare our algorithm with and without local search in Appendix~\ref{sec: appen ablation local}, where the additional gain from local search remains very minimal. These results support one of our main goals, which is substantially reducing the reliance on discrete local search procedures for exploring the solution space.}
\end{remark}

We refer the reader to the experiment of Appendix~\ref{sec: appen ablation global} that shows the impact of the choice of the global reset parameter $\rho$ in mQO.




\subsection{Relation to closely related algorithms}\label{sec: related work}

\textcolor{black}{Iterated Local Search (ILS)~\cite{lourencco2003iterated}, repeatedly perturbs a locally optimal solution and then reapplies discrete local search, effectively exploring the space of local optima rather than the full solution space directly. mQO shares this general perturb-and-improve approach, as both methods aim to escape locally trapped solutions and continue the search. A related similarity also exists with breakout-style methods such as~\cite{benlic2013breakout} (and some classical evolutionary algorithms \cite{back2018evolutionary}), which perturb the current solution to escape local optima. The main distinction is that mQO performs most of its exploration through continuous, differentiable optimization rather than purely discrete perturbation and local search. This enables exploration over the relaxed landscape, where stagnation is associated with stationary or fixed points of the continuous objective.}

Our algorithms are also related to the parallelized methods in~\cite{alkhouri2025differentiable,alkhouri2026MaxCut}, which similarly apply gradient updates to box-constrained quadratic objectives. The primary difference lies in how the search space is explored. In~\cite{alkhouri2025differentiable,alkhouri2026MaxCut}, a full random restart is performed after each batch, whereas mQO uses a partial reset that preserves part of the current solution while perturbing the remaining variables. As shown in Table~\ref{tab:er_comparison mis}, our strategy is substantially more effective in our experiments, with both the CPU-only and CPU+GPU variants outperforming the CPU+GPU pCQO-MIS method.

\begin{table*}[t]
\centering

\caption{\small{Main results across different ER graphs under time budget. 
Results are averaged over 8 random seeds, except for Maxcut graphs with ($n \ge 30000$), which use 3 seeds.
``--'' indicates that a method is unable to return a solution within the time limit. Bold (resp. underlined) entries indicate the highest (resp. second-highest) mean objective values.  \colorbox{gray!20}{Gray-shaded} (resp. \colorbox{blue!20}{blue-shaded}) methods use CPU (resp. CPU+GPU).}}
\vspace{-0.3cm}
\begin{subtable}[t]{\textwidth}
\centering
\caption{MIS results with a 5-minute budget.}
\label{tab:er_comparison mis}

\vspace{-0.1cm}

\begin{small}
\begin{sc}
\resizebox{0.98\textwidth}{!}{
\begin{tabular}{l|cccccc|cc}
\toprule
$(n, p)$ & \colorbox{gray!20}{Greedy-MIS} & \colorbox{gray!20}{CQO} & \colorbox{blue!20}{pCQO} & \colorbox{gray!20}{Gurobi} & \colorbox{gray!20}{ReduMIS} & \colorbox{blue!20}{RLSA} & \colorbox{gray!20}{mQO-MIS} & \colorbox{blue!20}{pmQO-MIS}\\
\midrule
$(1000, 0.10)$
& 55.00 $\pm$ 0.71
& 57.88 $\pm$ 0.93
& 59.25 $\pm$ 0.66
& 59.50 $\pm$ 0.71
& \textbf{67.00 $\pm$ 0.71}
& 65.25 $\pm$ 0.83
& 66.38 $\pm$ 0.70
& \underline{66.75 $\pm$ 0.78}\\

$(1000, 0.30)$
& 22.50 $\pm$ 0.50
& 23.13 $\pm$ 0.33
& 24.38 $\pm$ 0.48
& 22.12 $\pm$ 0.60
& \textbf{25.75 $\pm$ 0.43}
& 23.38 $\pm$ 0.48
& \underline{25.38 $\pm$ 0.48}
& \underline{25.38 $\pm$ 0.48}\\

$(1000, 0.50)$
& 14.00 $\pm$ 0.00
& 14.25 $\pm$ 0.43
& 14.75 $\pm$ 0.43
& 13.25 $\pm$ 0.43
& \textbf{15.00 $\pm$ 0.00}
& 12.88 $\pm$ 0.33
& 14.75 $\pm$ 0.43
& \underline{14.88 $\pm$ 0.33}\\
\midrule
$(3000, 0.03)$
& 158.88 $\pm$ 1.83
& 162.88 $\pm$ 1.80
& 164.25 $\pm$ 1.20
& 142.75 $\pm$ 1.09
& \underline{201.00 $\pm$ 0.86}
& 195.25 $\pm$ 1.09
& 198.25 $\pm$ 0.87
& \textbf{201.63 $\pm$ 1.88}\\

$(3000, 0.10)$
& 65.25 $\pm$ 1.10
& 67.75 $\pm$ 0.50
& 68.88 $\pm$ 0.60
& 54.38 $\pm$ 0.70
& \textbf{81.00 $\pm$ 0.50}
& 78.13 $\pm$ 0.78
& 79.63 $\pm$ 0.60
& \underline{80.88 $\pm$ 0.70}\\

$(3000, 0.33)$
& 23.00 $\pm$ 0.00
& 23.13 $\pm$ 0.33
& 24.38 $\pm$ 0.48
& 17.63 $\pm$ 0.70
& 25.88 $\pm$ 0.33
& 24.00 $\pm$ 0.71
& \underline{26.00 $\pm$ 0.48}
& \textbf{26.38 $\pm$ 0.70}\\
\midrule
$(10000, 0.50)$
& 16.88 $\pm$ 0.33
& 16.13 $\pm$ 0.60
& 17.13 $\pm$ 0.33
& 13.50 $\pm$ 0.87
& --
& 16.13 $\pm$ 0.33
& \underline{17.38 $\pm$ 0.48}
& \textbf{18.13 $\pm$ 0.48}\\

$(20000, 0.50)$
& 17.63 $\pm$ 0.48
& 16.38 $\pm$ 0.60
& 17.75 $\pm$ 0.43
& 14.00 $\pm$ 1.22
& --
& --
& \underline{18.25 $\pm$ 0.66}
& \textbf{19.00 $\pm$ 0.50}\\

$(30000, 0.50)$
& 18.13 $\pm$ 0.33
& 16.88 $\pm$ 0.78
& 17.88 $\pm$ 0.71
& 14.50 $\pm$ 0.86
& --
& --
& \underline{18.63 $\pm$ 0.48}
& \textbf{19.13 $\pm$ 0.33}\\

\bottomrule
\end{tabular}}
\end{sc}
\end{small}
\end{subtable}

\vspace{0.2cm}

\begin{subtable}[t]{\textwidth}
\centering
\caption{MaxCut results with a 10-minute budget as compared to BLS.}
\label{tab:er_comparison cut}

\vspace{-0.01cm}

\begin{small}
\begin{sc}
\resizebox{0.98\textwidth}{!}{
\begin{tabular}{l|cccccc|cc}
\toprule
$(n, p)$ & \colorbox{gray!20}{Greedy} & \colorbox{gray!20}{Gurobi} & \colorbox{gray!20}{BLS} & \colorbox{blue!20}{RLSA} & \colorbox{blue!20}{FEM} & \colorbox{blue!20}{ANYCSP} & \colorbox{gray!20}{mQO-MaxCut} & \colorbox{blue!20}{pmQO-MaxCut}\\
\midrule
$(100, 0.50)$  
& \underline{-4.50 \tiny{$\pm$ 15.47}}
& \textbf{0.00 \tiny{$\boldsymbol{\pm}$ 15.24}}
& \textbf{1432.63 \tiny{$\boldsymbol{\pm}$ 15.52}}
& \textbf{0.00 \tiny{$\boldsymbol{\pm}$ 15.52}}
& \textbf{0.00 \tiny{$\boldsymbol{\pm}$ 15.52}}
& \textbf{0.00 \tiny{$\boldsymbol{\pm}$ 15.52}}
& \textbf{0.00 \tiny{$\boldsymbol{\pm}$ 15.52}}
& \textbf{0.00 \tiny{$\boldsymbol{\pm}$ 15.52}} \\
\midrule
$(1000, 0.10)$ 
& -751.00 \tiny{$\pm$ 95.91}
& \textbf{+28.13 \tiny{$\boldsymbol{\pm}$ 97.53}}
& 28506.63 \tiny{$\pm$ 100.33}
& +18.88 \tiny{$\pm$ 98.57}
& \underline{+27.88 \tiny{$\pm$ 96.67}}
& +1.50 \tiny{$\pm$ 97.16}
& +11.75 \tiny{$\pm$ 99.61}
& +24.25 \tiny{$\pm$ 95.75} \\

$(1000, 0.50)$ 
& -1337.38 \tiny{$\pm$ 168.28}
& \textbf{+35.88 \tiny{$\boldsymbol{\pm}$ 173.39}}
& 130974.00 \tiny{$\pm$ 170.45}
& +15.75 \tiny{$\pm$ 170.04}
& +15.75 \tiny{$\pm$ 171.68}
& -1973.75 \tiny{$\pm$ 368.96}
& +23.38 \tiny{$\pm$ 169.66}
& \underline{+35.13 \tiny{$\pm$ 168.91}} \\

$(1000, 0.80)$ 
& -1117.00 \tiny{$\pm$ 98.52}
& \textbf{+46.63 \tiny{$\boldsymbol{\pm}$ 100.80}}
& 204749.13 \tiny{$\pm$ 106.35}
& +35.88 \tiny{$\pm$ 99.08}
& +8.00 \tiny{$\pm$ 99.20}
& -2185.63 \tiny{$\pm$ 236.17}
& +38.00 \tiny{$\pm$ 96.35}
& \underline{+40.13 \tiny{$\pm$ 97.36}} \\
\midrule
$(30000, 0.50)$ 
& -251951.67 \tiny{$\pm$ 6528.79}
& -432.33 \tiny{$\pm$ 5791.62}
& 113446464.00 \tiny{$\pm$ 5189.62}
& --
& -544757.33 \tiny{$\pm$ 22518.95}
& --
& \underline{+20249.00 \tiny{$\pm$ 7850.55}}
& \textbf{+21621.33 \tiny{$\boldsymbol{\pm}$ 7296.21}} \\

$(30000, 0.80)$ 
& -189256.33 \tiny{$\pm$ 3093.05}
& +18990.00 \tiny{$\pm$ 3118.89}
& 180745762.00 \tiny{$\pm$ 2678.16}
& --
& -480071.33 \tiny{$\pm$ 29050.80}
& --
& \underline{+29077.67 \tiny{$\pm$ 3123.32}}
& \textbf{+29806.00 \tiny{$\boldsymbol{\pm}$ 3074.53}} \\
\midrule
$(40000, 0.50)$
& -338841.00 \tiny{$\pm$ 8075.24}
& --
& \underline{201409886.67 \tiny{$\pm$ 5172.17}}
& --
& -846542.67 \tiny{$\pm$ 12601.64}
& --
& \textbf{+78474.00 \tiny{$\boldsymbol{\pm}$ 8723.07}}
& \textbf{+78474.00 \tiny{$\boldsymbol{\pm}$ 8723.07}} \\

$(40000, 0.80)$ 
& -276846.00 \tiny{$\pm$ 1788.26}
& --
& \underline{321133893.00 \tiny{$\pm$ 1790.13}}
& --
& -698149.00 \tiny{$\pm$ 14831.37}
& --
& \textbf{+58593.67 \tiny{$\boldsymbol{\pm}$ 4127.41}}
& \textbf{+58593.67 \tiny{$\boldsymbol{\pm}$ 4127.41}} \\
\bottomrule
\end{tabular}}
\end{sc}
\end{small}
\end{subtable}

\end{table*}

\section{Experimental Results}
\label{exp}

\textcolor{black}{\textbf{Settings}: Here, we select NetworkX ER graphs as this provides a more comprehensive test of scalability in terms of $n$ and $m$. Our evaluation covers graphs from 100 to 40,000 nodes and includes both sparse and dense settings (ranging from 2,500 to 639,984,000 edges). This range is larger and more diverse than many commonly used benchmarks.} Additional comparison results using benchmark datasets DIMACS \cite{johnson1996cliques}, and graphs with different structures such as Random satisfaction Benchmark (RB) from \cite{wu2025unrealized}, Barbasi-Albert (BA) \cite{barabasi1999emergence}, and Stochastic Block Models (SBMs) \cite{holland1983stochastic} (as graphs with community structures) are given in Appendix~\ref{sec: append additional comparisons}. For our mQO-MIS, pmQO-MIS, similar to \cite{alkhouri2025differentiable}, we use momentum gradient ascent (MGA). For mQO-MaxCut and pmQO-MaxCut, MGA is also used following the ablation study in Appendix~\ref{appen: ablation opt choice}. Momentum is set based on Appendix~\ref{appen: hyperparameter settings}. The reset parameter, $\rho$, is set according to Appendix~\ref{sec: appen ablation global}. The number of global reset rounds is set to $T_{\mathrm{gs}} = 60$ (resp. $T_{\mathrm{gs}} = 90$) for the MIS (resp. MaxCut) problem. For $\lambda$ in $f_\textrm{B}$, we set it to 0.001, following the ablation in Appendix~\ref{appen: lambda ablation}. For the large-scale ER graphs, the MaxCut formulation, $f_\textrm{B}$, uses sparse matrix representation for $\mathbf{A}$. All experiments and methods are run on a server equipped with a single CPU AMD EPYC 9565 processor. For the pmQO algorithms and the other GPU baselines, we combine the aforementioned CPU with an NVIDIA RTX PRO 6000 Blackwell GPU. Our code is available online\footnote{\tiny{\url{https://github.com/MrSunnyLight/Mutation-Guided-Differentiable-CO}}}.


\textcolor{black}{\textbf{Baselines \& selection justification:}} For MIS, we include \textbf{four CPU baselines}: (i) Greedy-MIS, which constructs solutions by iteratively selecting (at random) lower-degree nodes (which is used in many methods including \cite{chaplick2025approximation}) as this represents simple and an efficient greedy method; (ii) Clique-based Quadratic Optimization (CQO) \cite{alkhouri2025differentiable}, representing differentiable quadratic methods, where we restart $\mathbf{x}$ every time the optimizer is stuck (hyper-parameters are used following the paper's recommendation for different ER graphs); (iii) Gurobi \cite{Gurobi}, representing commercialized ILP to solve MIS ILP defined in Section~\ref{sec: prelim}; (iv) ReduMIS \cite{lamm2016finding}, representing a powerful heuristic that is widely regarded to return the best solutions on different settings and graphs including learning-based methods (see the recent study in \cite{wu2025unrealized}). We note that ReduMIS also utilizes the iterative $(1,2)$-swap local search from \cite{andrade2012fast}. Second, we include \textbf{two recent GPU-accelerated methods} that demonstrated very strong performance against many baselines including the training-data-intensive ones: (i) the parallelized (i.e., full) version of CQO, pCQO \cite{alkhouri2025differentiable}; (ii) Regularized Langevin dynamics with Simulated Annealing (RLSA) \cite{fengregularized} to represent recent powerful parallelized sampling methods. For MaxCut, we compare to \textbf{three CPU baselines}: (i) The well-known Greedy MaxCut algorithm \cite{williamson2011design}, representing approximation algorithms, where we fix an ordering on the vertices in $V$, then, we start with two empty bins $S$ and $\bar{S}$, and put $v_1$ in $S$, then, for each subsequent vertex, we place it in the bin such that $\texttt{Cut}(S)$ is maximized (other random orderings are used consecutively until the time budget expires); (ii) Gurobi, which is used to solve the ILP defined in Section~\ref{sec: prelim}; (iii) the BLS heuristic \cite{benlic2013breakout}, which is among the best known methods for solving the MaxCut problem as was recently demonstrated in \cite{zeng2024performance}. We note that BLS uses the iterative $1$-flip local search. Furthermore, we include \textbf{three GPU-accelerated baselines}: (i) RLSA MaxCut \cite{fengregularized}; (ii) the recent Free Energy Machine (FEM) \cite{shen2025free},inspired by free-energy minimization in statistical physics and achieves scalable results using relaxed QUBOs; (iii) ANYCSP \cite{tonshoff2023one} as a representative of an RL method that generalizes to multiple problems.

\begin{figure}[t]
    \centering
    \includegraphics[width=0.88\linewidth]{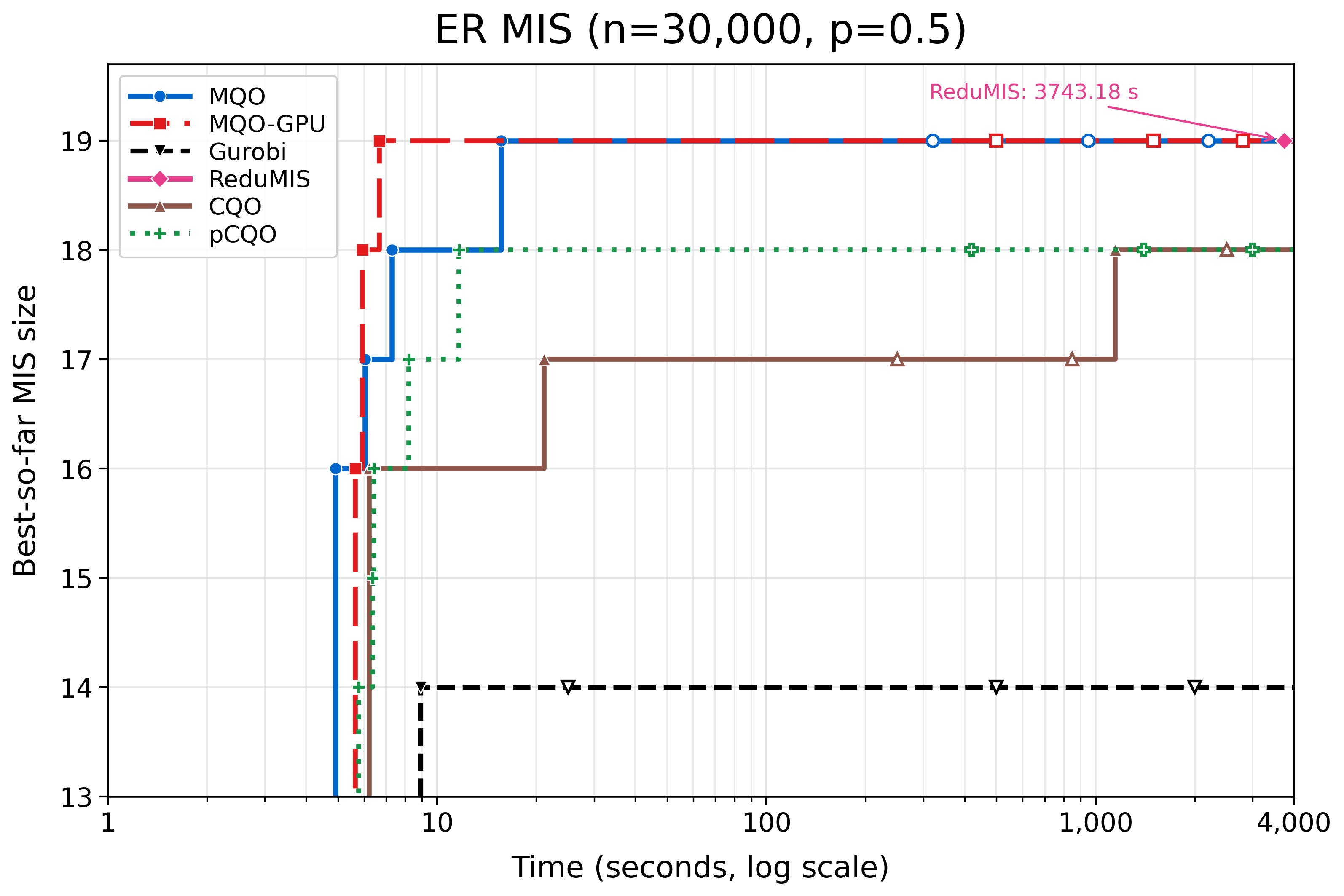}
    \vspace{-0.4cm}
    \caption{\small{Convergence over 4000 seconds for MIS.}}
    \label{fig:mis-best-so-far}
    \vspace{-0.5cm}
\end{figure}
\begin{figure}[t]
    \centering
    \includegraphics[width=0.88\linewidth]{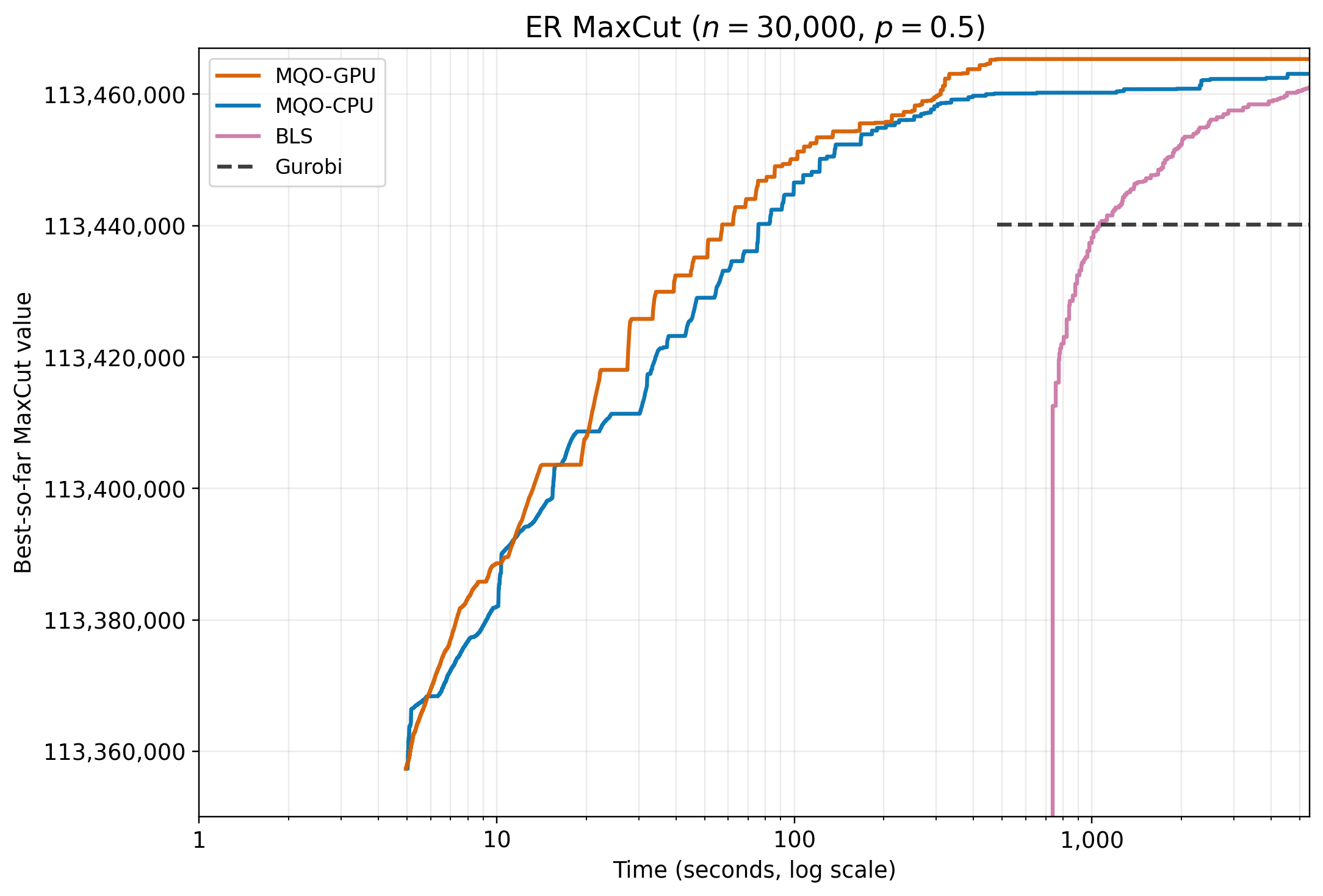}
    \vspace{-0.4cm}
    \caption{\small{Convergence over 5400 seconds for MaxCut.}}
    \label{fig:maxcut-best-so-far}
    \vspace{-0.5cm}
\end{figure}


\textbf{Main results under limited budget:} For MIS, Table~\ref{tab:er_comparison mis} presents the results under $5$ minutes. For smaller graphs, other than $(n,p) = (3000,0.03)$, ReduMIS returns the highest MIS sizes as compared to all CPU and GPU methods, with mQO and pmQO achieving the second highest. However, for the large graphs ($(n,d) \geq (3000,0.33)$), our GPU version, pmQO, achieves the best results, while mQO achieves the second highest. Furthermore, our CPU-only version, mQO, obtains better results than the two recent GPU methods, pCQO and RLSA, all without the need of parallelization, highlighting the scalable performance of Algorithm~\ref{alg: main mQO-MIS}. Table~\ref{tab:er_comparison cut} presents results for the MaxCut problem w.r.t. BLS (as its the SOTA local search based heuristic) with a time budget of $10$ minutes. Our methods achieve the best (on-par) results only for $(n,p) = (100,0.50)$, while our GPU version obtains the second-best results for ER graphs with $n\leq 1000$. In this regime, Gurobi achieves the best overall performance. Furthermore, our CPU-only version, $n=1000$ and $p\in \{0.5,0.8\}$, obtains better results than all the GPU baselines, all without parallelization. For large-scale graphs, our GPU and CPU version achieve the best and second-best results, demonstrating the scalability of mQO. 

\textcolor{black}{\textbf{Convergence:} Figure~\ref{fig:mis-best-so-far} and Figure~\ref{fig:maxcut-best-so-far} present convergence results over 4000 and 5400 seconds for MIS and MaxCut, respectively, using a large ER graph with $(n,p) = (30000,0.5)$. We choose the strongest baselines in these regimes. For MIS, mQO (both versions) achieves the best solution very quickly.
ReduMIS required 3743.18 seconds to obtain the solution. For MaxCut, our algorithms not only converge to the best solution, but also obtains the first solution in less than 10 seconds as compared to BLS and Gurobi who required more than 500 seconds.}

\begin{table}[htb]
\caption{\small{Breakdown of improvements from gradient updates, global reset, and discrete local search for MIS and MaxCut on ER graphs with $(n,p) = (3000,0.03)$ and $(n,p) = (40000, 0.80)$ respectively. Values in parentheses indicate increased values}.}
\vspace{-0.2cm}
\label{tab:progress_breakdown}
\centering
\resizebox{0.49\textwidth}{!}{
\begin{tabular}{lcc}
\toprule
\textbf{Stage} & \textbf{MIS} & \textbf{MaxCut} \\
\midrule
End of gradient updates (Step~3)    & 161      & 321035650 \\
After global reset (Step~7-9)          & 194 (33) & 321189058 (153408) \\
\midrule
After discrete local search (Step~11)  & 196 (2)  & 321189085 (27) \\
\midrule
Final value                 & 196      & 321189085 \\
\bottomrule
\end{tabular}
}
\end{table}

\noindent
\paragraph{mQO main gains are from differentiable updates:} \textcolor{black}{To support the claim that the main gains of mQO comes from differentiable updates as compared to discrete local search, we run mQO-MIS and report the value after the gradient updates, as well as the incremental gains obtained from the global reset (mutations) and local searches. 
Results of Table~\ref{tab:progress_breakdown} show that most of the improvement is driven by the gradient updates and the global reset, while the discrete local searches contribute only marginal gains. This indicates that mQO's performance primarily comes from continuous optimization and large-scale subspace exploration, with discrete search serving as a lightweight refinement.}



\section{Conclusion}

In this paper, we analyzed the limitations of gradient-based methods for relaxed box-constrained QUBO formulations of the Maximum Independent Set and Maximum Cut problems from differentiable repairability perspective, and showed that stalling at local maxima is a fundamental bottleneck. For MaxCut, we proposed a new formulation that provably has the ability to escape one-node repairable points. Motivated by these insights, we introduced mQO, a framework that combines differentiable optimization and stochastic mutation-based resets (along with optional discrete local search) to expand the search radius when optimization stalls. Our experiments show that our algorithms are particularly well-suited for obtaining fast solutions on large-scale graphs. Limitations and future work are discussed in Appendix~\ref{sec: append limiations and future work}.

\bibliography{AISTATS27_try/refs_aistats}
\bibliographystyle{abbrv}
\newpage
\appendix
\onecolumn
\begin{center}
{\Large \bf Appendix}
\end{center}
\vspace{-0.1in}
\par\noindent\rule{\textwidth}{1pt}


\section*{Table of Contents}


\begingroup
\newcommand{\appendixtocentry}[2]{%
    \noindent
    \makebox[2em][l]{\ref{#1}}%
    #2\dotfill\pageref{#1}\par
}

\appendixtocentry
    {appen: proofs}
    {Proofs}

\appendixtocentry
    {sec: append partial reset > random re-initialization}
    {Partial Reset vs. Full Random Restart}

\appendixtocentry
    {sec: append empirical support of aflip escap of fB}
    {Impact of $1$-Flip Escapability in MaxCut}

\appendixtocentry
    {sec: append cut obj comparison}
    {Comparison of MaxCut Formulations}

\appendixtocentry
    {sec: append additional comparisons}
    {Results on RB, BA, SBM, and DIMACS Graphs}

\appendixtocentry
    {sec: appen ablation global}
    {Impact of Global Reset and the Choice of its parameter $\rho$}

\appendixtocentry
    {sec: appen ablation local}
    {Impact of the Adopted Local Search}

\appendixtocentry
    {appen: ablation opt choice}
    {Choice of the optimizer for Maxcut}

\appendixtocentry
    {appen: implementation of pmQO}
    {mQO CPU and GPU Pseudocodes }

\appendixtocentry
    {appen: lambda ablation}
    {Impact of $\lambda$ in $f_{\mathrm{B}}$}

\appendixtocentry
    {appen: hyperparameter settings}
    {Hyperparameter Settings}

\appendixtocentry
    {sec: append limiations and future work}
    {Limitations and Future Work}

\endgroup

\section{Proofs}\label{appen: proofs}


\subsection{Proposition~\ref{prop: PGA is stuck at (1,2)}}
\begin{proof}
The proof follows from showing that for any maximal independent set, including the ones where the $(1,2)$-swap in Definition~\ref{def: 1,2 swap} applies, the gradient direction always points towards outside the boundary $[0,1]^n$, and by the projection operator, we get the same point.  

Let $I$ be a MaxIS and define binary vector $\mathbf{x}$ such that it contains ones only at indices corresponding to the nodes in set $I$. For any $I$ (including the $(1,2)$-swap MaxIS), we have 
%
\begin{equation}\label{eqn: der MIS}
   \frac{\partial h(\mathbf{x})}{\partial\mathbf{x}_v} \leq 0, \forall v\in I, ~~~\text{and}~~  \frac{\partial h(\mathbf{x})}{\partial\mathbf{x}_v} \geq 0, \forall v\notin I. 
\end{equation}
%
 The first equation is derived because if $v\notin I$, then $\mathbf{x}_v=0$ (by the definition of $\mathbf{x}$) so it is at the left boundary of the interval $[0,1]$. For the left boundary point to be a local maximizer, it requires the derivative to be non-negative (i.e., moving towards the right only increases the objective). Similarly, when $v\in I$, $\mathbf{x}_v=1$, is at the right boundary for \eqref{eqn: der MIS}, at which the derivative should be non-positive. 
 
The derivative of $h$ computed as 
\begin{equation}\label{eq:expre_f}
\frac{\partial h(\mathbf{x})}{ \partial \mathbf{x}_v}  = 1 - \gamma m_v, \quad \forall v \notin I,
\end{equation}
where
\begin{equation}\label{eq:mv}
m_v:= \left|\{u\in \mathcal{N}(v) \cap I \}\right|\:,
\end{equation}
is the number of neighbors of $v$ in $I$. By this definition, we immediately have $1\leq m_v \leq |I| $, where the upper and lower bounds for $m_v$ are all attainable by some special graphs. Note that the lower bound of $m_v$ is $1$, and that is due the fact that $I$ is a MaxIS, so any other node (say $v$) will have at least $1$ edge connected to a node in $I$. Plugging \eqref{eq:expre_f} into the first equation of \eqref{eqn: der MIS}, we obtain
%
\begin{equation} \label{eq:gamma_Q}
\gamma \geq \frac{1}{m_v}\:.
\end{equation} 
Since we have $\gamma>1$, we must set $m_v$ to its lowest possible value, $1$ (attainable by some graphs), and still requires $\gamma$ to satisfy \eqref{eq:gamma_Q}. This means that for any MaxIS $I$ and $\gamma>1$, the gradient direction for all nodes is towards infeasible points w.r.t. the box-constraints. \end{proof}






\subsection{Theorem~\ref{th: all fixed of perturbed laplacian}}

\begin{proof}
Write 
\[
\mathbf{z}:=\mathbf{x}+\alpha\nabla f_{\mathrm L}(\mathbf{x})=\mathbf{x}+\frac{\alpha}{2}\mathbf{L}\mathbf{x}\:.
\]
For any $i$ and any binary $\mathbf{x}$, using $\mathbf{L}=\mathbf{D}-\mathbf{A}$ with nonnegative weights, we obtain
\[
(\mathbf{L}\mathbf{x})_i=\sum_{j} \mathbf{A}_{ij}(\mathbf{x}_i-\mathbf{x}_j)=2\mathbf{x}_i\sum_{j:\,\mathbf{x}_j\neq \mathbf{x}_i}\mathbf{A}_{ij},
\]
which means that $\mathbf{x}_i(\mathbf{L}\mathbf{x})_i\ge 0$. 

If $\mathbf{x}_i=1$, then, 


\[\mathbf{z}_i=1+\frac{\alpha}{2}\bigl(\mathbf{L}\mathbf{x}\bigr)_i\ge 1,\] hence $\Pi_{[-1,1]}(\mathbf{z}_i)=1=\mathbf{x}_i$.

If $\mathbf{x}_i=-1$, then, $\mathbf{z}_i\le -1$, hence $\Pi_{[-1,1]}(\mathbf{z}_i)=-1=\mathbf{x}_i$.

Thus, $\Pi_{[-1,1]^n}(\mathbf{z})=\mathbf{x}$.
\end{proof}
%






\subsection{Proposition~\ref{prop: adj form is not 1-flip repair}}

\begin{proof}
For binary $\mathbf{x}$ and a flip index $i$, let \[\mathbf{x}^{(i)}=\mathbf{x}-2\mathbf{x}_i \mathbf{e}_i\] be the $1$-flip repairable point from Definition~\ref{def: 1-flip}. A direct expansion gives the $1$-flip gain:
\[
f_{\mathrm{\mathbf{A}}}(\mathbf{x}^{(i)})-f_{\mathrm{\mathbf{A}}}(\mathbf{x})
= -(\mathbf{x}-2\mathbf{x}_i\mathbf{e}_i)^T \mathbf{A}(\mathbf{x}-2\mathbf{x}_i\mathbf{e}_i)+\mathbf{x}^T \mathbf{A} \mathbf{x}
= 4\mathbf{x}_i(\mathbf{A}\mathbf{x})_i,
\]
using $\mathbf{A}_{ii}=0, \forall i\in V$. If $\mathbf{x}$ is $1$-flip repairable, then for some $i$ we have $4\mathbf{x}_i(\mathbf{A}\mathbf{x})_i>0$. But \[\nabla f_{\mathrm{\mathbf{A}}}(\mathbf{x})=-2\mathbf{A}\mathbf{x},\] so \[\mathbf{x}_i\nabla_i f_{\mathrm{\mathbf{A}}}(\mathbf{x})=-2\mathbf{x}_i(\mathbf{A}\mathbf{x})_i<0.\]

At a binary point, the PGA fixed-point conditions require $\mathbf{x}_i\nabla_i f(\mathbf{x})\ge 0$ for all $i$
(because at $\mathbf{x}_i=1$, one needs $\nabla_i f\ge 0$, and at $\mathbf{x}_i=-1$ one needs $\nabla_i f\le 0$).
Thus $\mathbf{x}$ cannot be a fixed point of PGA, hence $\Pi_{[-1,1]^n}(\mathbf{x}+\alpha \nabla f_{\mathrm{\mathbf{A}}}(\mathbf{x}))\neq \mathbf{x}$.
\end{proof}
%


\subsection{Proposition~\ref{prop: adj_fractional_localmin}}

    

\begin{proof}
Let $\mathbf{s} \in \{-1,1\}^n$ be a binary global maximizer of $f_{\mathbf{A}}$, $s_u$ denote the assignment of vertex $u$, and suppose $G$ has a vertex $v$ whose neighbors in $\mathbf{s}$ are split
evenly, i.e., \[\sum_{u \in N(v)} s_u = 0.\] Define $\mathbf{x}^{(t)}$ by
replacing $s_v$ with $t \in [-1,1]$ and leaving the other coordinates
unchanged. Since only edges incident to $v$ depend on this coordinate,
\[
f_{\mathbf{A}}(\mathbf{x}^{(t)}) - f_{\mathbf{A}}(\mathbf{s})
= -2(t - s_v)\sum_{u \in N(v)} s_u
= 0,
\]
so $f_{\mathbf{A}}$ is constant in $t$ along this segment. Multilinearity of $f_{\mathbf{A}}$ implies that its maximum over $[-1,1]^n$ is attained at a binary vertex of the box, so $\mathbf{s}$ is also a global maximizer over $[-1,1]^n$, and hence so is every $\mathbf{x}^{(t)}$. For any $t \in (-1,1)$, $\mathbf{x}^{(t)}$ is a non-binary global maximizer of $f_{\mathrm A}$ on $[-1,1]^n$ that does not correspond to a binary MaxCut solution.

Such balanced vertices can be constructed for different graphs. For example, $G = K_3$ with binary
maximizer $\mathbf{s} = (1, 1, -1)$, vertex $1$ satisfies $s_2 + s_3 = 0$, so $\mathbf{x}^{(t)} = (t, 1, -1)$ is a non-binary global maximizer for every $t \in (-1, 1)$.
\end{proof}







\subsection{Theorem~\ref{thm: adj_lin_two_props}}

\begin{proof}
(i) Let $\mathbf{x}\in\{-1,1\}^n$ and $\mathbf{x}^{(i)}=\mathbf{x}-2\mathbf{x}_i\mathbf{e}_i$. Since $\mathbf{A}_{ii}=0$, we have
\[
f_{\mathrm A}(\mathbf{x}^{(i)})-f_{\mathrm A}(\mathbf{x})
= 4\mathbf{x}_i(\mathbf{A}\mathbf{x})_i.
\]
Thus, $f_{\mathrm A}(\mathbf{x}^{(i)})>f_{\mathrm A}(\mathbf{x})$ implies that $\mathbf{x}_i(\mathbf{A}\mathbf{x})_i>0$ for some $i$.

For the perturbed biased objective, we have $\nabla f_{\mathrm{B}}(\mathbf{x})=-2\mathbf{A}\mathbf{x}-\lambda \mathbf 1$, hence
\[
\mathbf{x}_i \nabla_i f_{\mathrm{B}}(\mathbf{x})=-2\mathbf{x}_i(\mathbf{A}\mathbf{x})_i-\lambda \mathbf{x}_i.
\]
If $\mathbf{x}_i(\mathbf{A}\mathbf{x})_i>0$, then $2\mathbf{x}_i(\mathbf{A}\mathbf{x})_i\geq 2$ (it is an even integer), so
\[
\mathbf{x}_i \nabla_i f_{\mathrm{B}}(\mathbf{x})\le -2+\lambda <0
\qquad(\text{for any } \lambda\in(0,2)).
\]
At a binary point, being a PGA fixed point requires $\mathbf{x}_i\nabla_i f_{\mathrm{B}}(\mathbf{x})\ge 0$ for all $i$
(equivalently, $\nabla_i f_{\mathrm{B}}(\mathbf{x})\ge 0$ when $\mathbf{x}_i=1$ and $\nabla_i f_{\mathrm{B}}(\mathbf{x})\le 0$ when $\mathbf{x}_i=-1$).
Therefore, $\mathbf{x}$ cannot be a fixed point of PGA.

\medskip
(ii) Let $\mathbf{x}^\star$ be a local maximizer of $f_{\mathrm{B}}$ over $[-1,1]^n$ and define the interior index set
$S:=\{i:\ |\mathbf{x}_i^\star|<1\}$.  For each $i\in S$, first-order necessary conditions give
\[
\nabla_i f_{\mathrm{B}}(\mathbf{x}^\star)=0
\quad\Longleftrightarrow\quad
2(\mathbf{A}\mathbf{x}^\star)_i+\lambda=0
\]
\[
\quad\Longleftrightarrow\quad
(\mathbf{A}\mathbf{x}^\star)_i=-\lambda/2.
\tag{$\star$}
\]
Suppose $S\neq\emptyset$. Consider the induced subgraph on $S$ with adjacency matrix $\mathbf{A}_{SS}$.

\emph{Case 1: $\mathbf{A}_{SS}=0$ (no edges inside $S$).}
Then each $(\mathbf{A}\mathbf{x}^\star)_i$ for $i\in S$ is a sum of $\pm1$ terms coming only from neighbors in $V\setminus S$
(because $S$ is independent and $\mathbf{x}^\star_j=\pm1$ for $j\notin S$). Hence, $(\mathbf{A}\mathbf{x}^\star)_i\in\mathbb Z$.
This contradicts $(\star)$ since $\lambda/2\notin\mathbb Z$. Thus, this case is impossible.

\emph{Case 2: $\mathbf{A}_{SS}\neq 0$ (there is at least one edge inside $S$).}
Then $\mathbf{A}_{SS}$ has at least one positive eigenvalue (e.g. by Perron--Frobenius on a nonnegative matrix),
and $\mathrm{tr}(\mathbf{A}_{SS})=0$ (zero diagonal), so it must also have a negative eigenvalue:
$\lambda_{\min}(\mathbf{A}_{SS})<0$. Let $d\in\mathbb R^n$ be supported on $S$ with $\mathbf{d}_S$ a unit eigenvector
for $\lambda_{\min}(\mathbf{A}_{SS})$. For sufficiently small $t>0$, 
\[\mathbf{x}^\star+t \mathbf{d}\in \mathcal U[-1,1]^n,\] and
a second-order expansion yields
\[
-f_{\mathrm{B}}(\mathbf{x}^\star+t \mathbf{d})+f_{\mathrm{B}}(\mathbf{x}^\star)
= -t\,\nabla f_{\mathrm{B}}(\mathbf{x}^\star)^T \mathbf{d} + t^2\, \mathbf{d}^T \mathbf{A} \mathbf{d}.
\]
Because $d$ is supported on $S$ and $\nabla_i f_{\mathrm{B}}(\mathbf{x}^\star)=0$ for $i\in S$, the linear term vanishes.
Moreover,
\[
\mathbf{d}^T \mathbf{A} \mathbf{d} = \mathbf{d}_S^T \mathbf{A}_{SS} \mathbf{d}_S = \lambda_{\min}(\mathbf{A}_{SS}) < 0,
\]
so for small $t$, we obtain $f_{\mathrm{B}}(\mathbf{x}^\star+t \mathbf{d}) > f_{\mathrm{B}}(\mathbf{x}^\star)$, contradicting local maximality.

Thus, $S$ must be empty, meaning that every coordinate of $\mathbf{x}^\star$ lies on $\{-1,1\}$, i.e.\ $\mathbf{x}^\star\in\{-1,1\}^n$.
\end{proof}

\section{Partial Reset vs. Full Random Restart}\label{sec: append partial reset > random re-initialization}


\textcolor{black}{In order to show the advantage of the partial reset in mQO, we compare CQO with full random restart vs. our algorithm without local search (last step). This is because CQO does not adopt local search}. 

Table~\ref{tab:partial-reset-vs-cqo} presents the results. As observed, mQO-MIS without local search achieves larger mean MIS sizes than CQO in all three graph settings. These results support the effectiveness of our partial-reset framework.

\begin{table}[h]
    \centering
    \small
    \caption{
    Comparison between CQO with random restart and mQO-MIS without local search on ER graphs. Results are reported as mean $\pm$ standard deviation. Bold and underlined values denote the highest and second-highest results, respectively.
    }
    \label{tab:partial-reset-vs-cqo}
    \begin{tabular}{ccc}
        \toprule
        $(n,p)$ 
        & \colorbox{gray!20}{CQO} (a representative of ``full random restart'')
        & \colorbox{gray!20}{mQO-MIS} (w/out local search \& our ``partial reset'') \\
        \midrule

        $(3000, 0.03)$ 
        & \underline{162.88 $\pm$ 1.80}
        & \textbf{197.50 $\pm$ 0.93} \\

        $(3000, 0.10)$ 
        & \underline{67.75 $\pm$ 0.50}
        & \textbf{79.50 $\pm$ 0.76} \\

        $(3000, 0.33)$ 
        & \underline{23.38 $\pm$ 0.33}
        & \textbf{25.88 $\pm$ 0.35} \\

        \bottomrule
    \end{tabular}
\end{table}

\section{Theorem~\ref{thm: adj_lin_two_props} Empirical Support: Impact of $1$-Flip Escapability}\label{sec: append empirical support of aflip escap of fB}

Here, we present a study that empirically support the impact of the $1$-flip escapability of $f_\textrm{B}$ in Theorem~\ref{thm: adj_lin_two_props}. More specifically, we ask the question: \textit{Would using $f_\textrm{B}$ reduce the dependence on discrete local search?} We answer in the affirmative. The results are in Table~\ref{tab: local_search_improvement}. 

When the standard Laplacian objective $f_\textrm{L}$ is used, the method relies heavily on local search to obtain improvements, because PGA stalls at many locally repairable binary points. In contrast, the proposed objective, i.e., $f_\textrm{B}$, is designed so that PGA has the ability to escape $1$-flip repairable points and returns a substantially better solution before local search is applied. Thus, the new formulation directly reduces the amount of discrete local-search work required.


\begin{table}[h]
\centering
\caption{Impact of discrete local search (LS) when we use the standard formulation $f_\textrm{L}$ vs. our developed one $f_\textrm{B}$.}
\label{tab: local_search_improvement}
\begin{small}
\resizebox{0.92\textwidth}{!}{
\begin{tabular}{lcccc}
\toprule
\textbf{Objective} & 
\textbf{Avg. cut before LS} & 
\textbf{Avg. cut after LS} & 
\textbf{Improvement from LS} & 
\textbf{Percentage} \\
\midrule
mQO-MaxCut with $f_\textrm{B}$ (default) 
& $28425.68 \pm 35.83$ 
& $28445.55 \pm 35.66$ 
& $19.87 \pm 0.22$ 
& $0.0699$ \\

mQO-MaxCut with $f_\textrm{L}$ 
& $27350.88 \pm 35.53$ 
& $28432.24 \pm 34.82$ 
& $1081.35 \pm 2.07$ 
& $3.9536$ \\
\bottomrule
\end{tabular}}
\end{small}
\vspace{-0.3em}
\end{table}




\section{Comparison of MaxCut Formulations}\label{sec: append cut obj comparison}

Here, we report the average cut value over 3 graphs for each $(n,d)$, comparing the three discussed objectives in Section~\ref{sec: theory}. 
We run the relaxed QUBOs without local search and without global reset, and starting from the same exact initialization. 

\vspace{0.3em}
\begin{table}[htp!]
\centering
\caption{Overall comparison between the discussed MaxCut formulations across different ER graphs as indicated by $(n,d)$ of the first row.}
\label{tab: cut formualtions comparison}
\begin{small}
\resizebox{0.93\textwidth}{!}{
\begin{tabular}{lcccc}
\toprule
\textbf{MaxCut Formulation} & $(100,50)$ & $(1000,100)$ & $(1000,500)$ & $(1000,800)$ \\
\midrule
Laplacian $f_\textrm{L}$ \eqref{eqn: Cut QUBO}         & 1256.2  & 25298.11 & 125046.6 & 200030.6 \\
Adjacency $f_\textrm{A}$    \eqref{eqn: adj cut formualtion}        & 1397.16 & 28157.02 & 130331.9 & 204291.2 \\
Perturbed Biased $f_\textrm{B}$ (the one used in our Algorithms)   \eqref{eqn: adj_lin_two_props}             & \textbf{1397.45} & \textbf{28166.59} & \textbf{130336.5} & \textbf{204296.5} \\
\bottomrule
\end{tabular}}
\end{small}
\vspace{-0.3em}
\end{table}

As observed in Table~\ref{tab: cut formualtions comparison}, $f_\textrm{A}$ and $f_\textrm{B}$ achieve significantly higher cut values when compared to $f_\textrm{L}$. 
Also, $f_B$ (the one used in Algorithms~\ref{alg: main mQO-MIS} and \ref{alg: main mQO-Cut}), characterized by 
Theorem~\ref{thm: adj_lin_two_props}, achieves slightly higher cut values than $f_A$.

\section{Results on RB, BA, SBM, and DIMACS Graphs}\label{sec: append additional comparisons}

Our main claim is that our method achieves SOTA performance or is very competitive on larger scale graphs (in size and density) under short time budgets and when compared to SOTA methods: ReduMIS, BLS, and Gurobi. This claim is supported by our large ER results in Tables~\ref{tab:er_comparison mis} and \ref{tab:er_comparison cut}.

To further demonstrate our claim of superiority in the large scale regime, here, we provide CPU comparison using stochastic block models (SBMs) graphs (as examples of graphs with community structures, generated using the NetworkX library), CPU+GPU comparison using small and large Random satisfaction Benchmark (RB generated using the code base in \cite{fengregularized}) and Barabasi-Albert (BA, generated using the NetworkX library) graphs, and CPU comparison on the well-known DIAMCS benchmark \cite{johnson1996cliques}. In the tables of this section, \colorbox{gray!20}{Gray-} (resp. \colorbox{blue!20}{blue-}) shaded methods use a CPU (resp. CPU + GPU). 

SBM graphs are characterized by $n$ (number of nodes), $k$ (number of clusters), and $p_{in} > p_{out}$ which are the probability of intra-edge creation and the probability of inter-edge creation, respectively. For RB graphs, they are characterized by $n = n'p'$ (where $n'$ is the number of partitions and $p'$ is the number of nodes contained in each partition), and which is the average degree of the graph. For Barabasi-Albert (BA) graphs, they are characterized by $n$ and $m'$ which is the number of edges that each new node attaches. We note that, for RB and BA, the smallest graphs ($n=3000$) we test here is more than twice as large as the largest that was evaluated in the RLSA GPU paper ($n=1200$) \cite{fengregularized}.

We note that for graphs with multiple connected components, i.e., graphs in which a path between every pair of nodes may not exist, each component can be treated independently, since the overall MIS (or MaxCut) can be obtained by concatenating the sets from the individual components. Furthermore, isolated nodes (i.e., any $v$ with $\textrm{d}(v) = 0$) can be initially removed from the graph and then appended to the final MIS. For MaxCut, isolated nodes have no effect to the final cut value as they are edge-less.

We use a time limit of 5 minutes for MIS and 10 minutes for MaxCut. For each graph setting, the reported results are averaged over 3 random seeds.

\begin{table}[h]
\centering
\caption{MIS on SBMs (2 clusters and $n=5000$).}
\label{tab:sbm_mis}
\resizebox{0.55\textwidth}{!}{
\begin{tabular}{l|cc|c}
\toprule
$(p_{in}, p_{out})$ & \colorbox{gray!20}{Gurobi} & \colorbox{gray!20}{ReduMIS} & \colorbox{gray!20}{mQO-MIS} (Ours) \\
\midrule
(0.10, 0.05) & 80.67 & \textbf{111} & 107.67 \\
(0.20, 0.05) & 51.33 & \textbf{70} & 67.33 \\
(0.30, 0.05) & 33.67 & \textbf{50} & 48.33 \\
(0.40, 0.05) & 25.33 & 34.67 & \textbf{37} \\
(0.50, 0.05) & 22.33 & 26.33 & \textbf{29.33} \\
\bottomrule
\end{tabular}}
\end{table}

\begin{table}[h]
\centering
\caption{MaxCut on SBMs (2 clusters and $n=30000$).}
\label{tab:sbm_maxcut}
\resizebox{0.77\textwidth}{!}{
\begin{tabular}{l|ccc|c}
\toprule
$(p_{in}, p_{out})$ & \colorbox{gray!20}{Greedy} & \colorbox{gray!20}{Gurobi} & \colorbox{gray!20}{BLS} & \colorbox{gray!20}{mQO-MaxCut} (Ours) \\
\midrule
(0.10, 0.05) & 17243748 & 17374698.67 & \textbf{17384884.33} & 17384622.67 \\
(0.30, 0.05) & 39879834 & 40059772.33 & 40067676.67 & \textbf{40075474.67} \\
(0.50, 0.05) & 62414402.33 & 62606129 & 62614216.33 & \textbf{62624312} \\
(0.70, 0.05) & 84878154.67 & 85060107.33 & 85060616 & \textbf{85075080} \\
(0.90, 0.05) & 107239124.33 & 107373206 & 107372057.67 & \textbf{107383285.33} \\
\bottomrule
\end{tabular}}
\end{table}

\begin{table}[h]
\centering
\caption{MIS on larger RB graphs.}
\label{tab:rb_mis}
\resizebox{0.92\textwidth}{!}{
\begin{tabular}{l|ccc|cc}
\toprule
RB $(n = n'p', d)$ & \colorbox{gray!20}{Gurobi} & \colorbox{gray!20}{ReduMIS} & \colorbox{blue!20}{RLSA} (GPU only) & \colorbox{gray!20}{mQO-MIS} & \colorbox{blue!20}{pmQO-MIS} \\
\midrule
(3000 = 30$\times$100, 1500) & 13.33 & 16 & 16.33 & 16.33 & \textbf{17} \\
(10000 = 100$\times$100, 5000) & 12.67 & -- & 15.67 & 18 & \textbf{18.67} \\
\bottomrule
\end{tabular}}
\end{table}

\begin{table}[htbp]
\centering
\caption{MaxCut on large BA graphs.}
\label{tab:ba_maxcut}
\resizebox{0.93\textwidth}{!}{
\begin{tabular}{l|cccc|cc}
\toprule
BA $(n, m')$ & \colorbox{gray!20}{Greedy} & \colorbox{gray!20}{Gurobi} & \colorbox{gray!20}{BLS} & \colorbox{blue!20}{RLSA} (GPU only) & \colorbox{gray!20}{mQO-MaxCut} & \colorbox{blue!20}{pmQO-MaxCut} \\
\midrule
(3000, 1500) & 1422866.67 & \textbf{1425100.33} & \textbf{1425100.33} & \textbf{1425100.33} & \textbf{1425100.33} & \textbf{1425100.33} \\
(30000, 15000) & 142110222 & -- & -- & -- & 142183221.33 & \textbf{142183224.67} \\
\bottomrule
\end{tabular}}
\end{table}

As observed, for the SBM results of Table~\ref{tab:sbm_mis} and Table~\ref{tab:sbm_maxcut}, mQO under-performs on sparser graphs, i.e., when $p_{in} < 0.4$ for MIS and $p_{in} < 0.3$ for MaxCut. However, as the edge density inside the clusters increases, our method clearly demonstrates stronger results as compared to Gurobi and SOTA heuristics.

For the RB graphs in Table \ref{tab:rb_mis}, when $n = 3000$, our GPU version (pmQO-MIS) finds the best results, whereas our CPU-only version is on par with both RLSA and ReduMIS. For the BA graphs of Table~\ref{tab:ba_maxcut}, our method is on par with the smaller BA graph. For larger BA graphs, not only our GPU version outperforms RLSA, but also our non-GPU version.

Furthermore, in Table~\ref{tab:dimacs_mis}, we compare mQO against Gurobi and SOTA heuristics using the DIMACS dataset (64 graphs) for MIS. We note that the largest DIMACS graph (with $n=4000$ and densoty around $0.5$) are generally smaller when compared to the large ER graphs in Table~\ref{tab:er_comparison mis}, the SBM graphs in Table~\ref{tab:sbm_mis}, and the RB graphs in Table~\ref{tab:rb_mis}. As observed, across different sizes and densities, our CPU version is on par with the SOTA heuristic, ReduMIS, on 61 out of 64 graphs. Also, our mQO-MIS achieves better results than Gurobi on five graphs. 



\begin{table*}[htbp]
\centering
\caption{DIMACS benchmark results for MIS.}
\label{tab:dimacs_mis}
\resizebox{0.6\textwidth}{!}{
\begin{tabular}{lccc|cc|c}
\toprule
Graph Name & $n$ & $m$ & Density & \colorbox{gray!20}{Gurobi} & \colorbox{gray!20}{ReduMIS} & \colorbox{gray!20}{mQO-MIS} \\
\midrule
brock200-1 & 200 & 14834 & 0.745427 & \textbf{6} & \textbf{6} & \textbf{6} \\
brock200-2 & 200 & 9876 & 0.496281 & \textbf{11} & \textbf{11} & \textbf{11} \\
brock200-3 & 200 & 12048 & 0.605427 & \textbf{9} & \textbf{9} & \textbf{9} \\
brock200-4 & 200 & 13089 & 0.657739 & \textbf{8} & \textbf{8} & \textbf{8} \\
brock400-1 & 400 & 59723 & 0.748409 & \textbf{7} & \textbf{7} & \textbf{7} \\
brock400-2 & 400 & 59786 & 0.749198 & \textbf{8} & \textbf{8} & 7 \\
brock400-3 & 400 & 59681 & 0.747882 & \textbf{7} & \textbf{7} & \textbf{7} \\
brock400-4 & 400 & 59765 & 0.748935 & \textbf{7} & \textbf{7} & \textbf{7} \\
brock800-1 & 800 & 207505 & 0.649265 & 9 & \textbf{10} & \textbf{10} \\
brock800-2 & 800 & 208166 & 0.651333 & 9 & \textbf{10} & \textbf{10} \\
brock800-3 & 800 & 207333 & 0.648727 & 10 & \textbf{11} & 10 \\
brock800-4 & 800 & 207643 & 0.649696 & \textbf{10} & \textbf{10} & \textbf{10} \\
C125-9 & 125 & 6963 & 0.898452 & \textbf{4} & \textbf{4} & \textbf{4} \\
C250-9 & 250 & 27984 & 0.899084 & \textbf{5} & \textbf{5} & \textbf{5} \\
C500-9 & 500 & 112332 & 0.900457 & \textbf{5} & \textbf{5} & \textbf{5} \\
C1000-9 & 1000 & 450079 & 0.901059 & \textbf{6} & \textbf{6} & 5 \\
C2000-5 & 2000 & 999836 & 0.500168 & 14 & \textbf{16} & \textbf{16} \\
C2000-9 & 2000 & 1799532 & 0.900216 & \textbf{6} & \textbf{6} & \textbf{6} \\
C4000-5 & 4000 & 4000268 & 0.500159 & 11 & \textbf{17} & \textbf{17} \\
c-fat200-1 & 200 & 1534 & 0.077085 & \textbf{18} & \textbf{18} & \textbf{18} \\
c-fat200-2 & 200 & 3235 & 0.162563 & \textbf{9} & \textbf{9} & \textbf{9} \\
c-fat200-5 & 200 & 8473 & 0.425779 & \textbf{3} & \textbf{3} & \textbf{3} \\
c-fat500-1 & 500 & 4459 & 0.035743 & \textbf{40} & \textbf{40} & \textbf{40} \\
c-fat500-2 & 500 & 9139 & 0.073259 & \textbf{20} & \textbf{20} & \textbf{20} \\
c-fat500-5 & 500 & 23191 & 0.185900 & \textbf{8} & \textbf{8} & \textbf{8} \\
c-fat500-10 & 500 & 46627 & 0.373764 & \textbf{4} & \textbf{4} & \textbf{4} \\
DSJC500-5 & 500 & 62624 & 0.501996 & \textbf{13} & \textbf{13} & \textbf{13} \\
DSJC1000-5 & 1000 & 249826 & 0.500152 & 13 & \textbf{15} & \textbf{15} \\
gen200-p0-9-44 & 200 & 17910 & 0.900000 & \textbf{5} & \textbf{5} & \textbf{5} \\
gen200-p0-9-55 & 200 & 17910 & 0.900000 & \textbf{5} & \textbf{5} & \textbf{5} \\
gen400-p0-9-55 & 400 & 71820 & 0.900000 & \textbf{8} & \textbf{8} & \textbf{8} \\
gen400-p0-9-65 & 400 & 71820 & 0.900000 & \textbf{7} & \textbf{7} & \textbf{7} \\
gen400-p0-9-75 & 400 & 71820 & 0.900000 & \textbf{6} & \textbf{6} & \textbf{6} \\
hamming6-2 & 64 & 1824 & 0.904762 & \textbf{2} & \textbf{2} & \textbf{2} \\
hamming6-4 & 64 & 704 & 0.349206 & \textbf{12} & \textbf{12} & \textbf{12} \\
hamming8-2 & 256 & 31616 & 0.968627 & \textbf{2} & \textbf{2} & \textbf{2} \\
hamming8-4 & 256 & 20864 & 0.639216 & \textbf{16} & \textbf{16} & \textbf{16} \\
hamming10-2 & 1024 & 518656 & 0.990225 & \textbf{2} & \textbf{2} & \textbf{2} \\
hamming10-4 & 1024 & 434176 & 0.828935 & \textbf{20} & \textbf{20} & \textbf{20} \\
johnson8-2-4 & 28 & 210 & 0.555556 & \textbf{7} & \textbf{7} & \textbf{7} \\
johnson8-4-4 & 70 & 1855 & 0.768116 & \textbf{5} & \textbf{5} & \textbf{5} \\
johnson16-2-4 & 120 & 5460 & 0.764706 & \textbf{15} & \textbf{15} & \textbf{15} \\
johnson32-2-4 & 496 & 107880 & 0.878788 & \textbf{31} & \textbf{31} & \textbf{31} \\
keller4 & 171 & 9435 & 0.649123 & \textbf{15} & \textbf{15} & \textbf{15} \\
keller5 & 776 & 225990 & 0.751546 & \textbf{31} & \textbf{31} & \textbf{31} \\
keller6 & 3361 & 4619898 & 0.818191 & 61 & \textbf{63} & \textbf{63} \\
MANN-a9 & 45 & 918 & 0.927273 & \textbf{3} & \textbf{3} & \textbf{3} \\
MANN-a27 & 378 & 70551 & 0.990148 & \textbf{3} & \textbf{3} & \textbf{3} \\
MANN-a45 & 1035 & 533115 & 0.996300 & \textbf{3} & \textbf{3} & \textbf{3} \\
MANN-a81 & 3321 & 5506380 & 0.998825 & \textbf{3} & \textbf{3} & \textbf{3} \\
p-hat300-3 & 300 & 33390 & 0.744482 & \textbf{9} & \textbf{9} & \textbf{9} \\
p-hat500-2 & 500 & 62946 & 0.504577 & \textbf{36} & \textbf{36} & \textbf{36} \\
p-hat500-3 & 500 & 93800 & 0.751904 & \textbf{10} & \textbf{10} & \textbf{10} \\
p-hat700-1 & 700 & 60999 & 0.249332 & \textbf{65} & \textbf{65} & \textbf{65} \\
p-hat700-2 & 700 & 121728 & 0.497560 & \textbf{49} & \textbf{49} & \textbf{49} \\
p-hat700-3 & 700 & 183010 & 0.748048 & \textbf{10} & \textbf{10} & \textbf{10} \\
p-hat1000-1 & 1000 & 122253 & 0.244751 & \textbf{75} & \textbf{75} & \textbf{75} \\
p-hat1000-2 & 1000 & 244799 & 0.490088 & \textbf{54} & \textbf{54} & \textbf{54} \\
p-hat1500-1 & 1500 & 284923 & 0.253434 & 85 & \textbf{87} & \textbf{87} \\
p-hat1500-2 & 1500 & 568960 & 0.506080 & \textbf{62} & \textbf{62} & \textbf{62} \\
p-hat1500-3 & 1500 & 847244 & 0.753608 & 9 & \textbf{12} & \textbf{12} \\
\bottomrule
\end{tabular}}
\end{table*}

Overall, these results further demonstrate the clear advantage of mQO in larger (and denser) graphs, even when compared to recent efficient GPU methods.



\section{Impact of Global Reset and the Choice of its parameter $\rho$}\label{sec: appen ablation global}

Here, we conduct two experiments to investigate the impact of the differentiable mutation-based strategy and the choice of the reset parameter, $\rho$. Table~\ref{tab:ablation_global_mis} and Table~\ref{tab:ablation_global_maxcut} report the results for the mQO-MIS and mQO-MaxCut algorithms, respectively, using values of $\rho \in [0,1)$. For both experiments, we use graphs with different number of nodes and density (as indicated by $(n,d)$). 

For both problems, the results in the first row (i.e., when $\rho=0$) correspond to running our algorithms \textit{without global reset}. The first main observation is that mQO without the global reset underperforms most of the cases $\rho > 0$. This highlights the importance of the proposed global reset strategy and its role in effectively improving performance. The second observation is that, for MIS, the best results are obtained when $\rho$ is set to $0.5$ or $0.6$, whereas for MaxCut, the best results are achieved when $\rho = 0.8$. Therefore, we use these values in our main experiments.

\section{Impact of The Adopted Local Search}\label{sec: appen ablation local}
\begin{table*}[htp]
      \caption{Ablation study on the impact of local search. For both problems, results are averaged over $8$ random seeds. Our algorithms here are run with including global reset.}
      \vspace{-0.1cm}
      \label{tab:ablation_local}
      \begin{center}
        \begin{small}
          \begin{sc}
          \resizebox{0.94\textwidth}{!}{
            \begin{tabular}{lccccc}
              \toprule
               Problem & Time Limit & $(n, d)$ & mQO without Local Search & mQO with Local Search  \\
              \midrule
              &&$(1000, 100)$  & 66.25  & \textbf{66.375}  \\
              &&$(1000, 300)$  & 25.25  & \textbf{25.375}  \\
              &&$(1000, 500)$  & \textbf{14.75}  & \textbf{14.75}  \\
              \textbf{MIS} & 5 Min
              &$(3000, 100)$  & 197.50 & \textbf{198.25} \\
              &&$(3000, 300)$  & 79.50  & \textbf{79.625}  \\
              &&$(3000, 1000)$ & \textbf{26.00}  & \textbf{26.00}  \\
              \midrule
        && $(100, 50)$      & \textbf{1432.625} & \textbf{1432.625} \\
        && $(1000, 100)$    & 28504.25   & \textbf{28518.375} \\
        \textbf{MaxCut} & 10 Min
        & $(1000, 500)$  & 130928.50  & \textbf{130997.375}  \\
        && $(1000, 800)$  & 204759.875 & \textbf{204787.125} \\

          \bottomrule
        \end{tabular}}
      \end{sc}
    \end{small}
  \end{center}
  \vskip -0.1in
\end{table*}

Here, we report an ablation study that highlights the impact of incorporating local search in mQO-MIS and mQO-MaxCut (last steps in Algorithms~\ref{alg: main mQO-MIS} and \ref{alg: main mQO-Cut}). Table~\ref{tab:ablation_local} presents the results for different graph orders and densities, as indicated in the third column. As observed, for both problems, including local search in mQO marginally improves the results. This holds for all configurations except in the third and last case for MIS and for $(n,d) = (100,50)$ in the MaxCut. Compared to the global reset, local search yields smaller improvements. This empirically motivates our proposed mutation-based reset in mQO.

\section{Choice of the optimizer For MaxCut}\label{appen: ablation opt choice}

In this section, we conduct a study, comparing vanilla gradient ascent (GA) with the momentum-based GA (MGA) we use in this paper. We use $\alpha = 0.001$ for GA and $\alpha = 0.0025$ for MGA where and the momentum parameter is set to $0.8$.  

    \begin{table}[htb]
      \caption{Ablation study comparing GA and momentum-based GA for the MaxCut problem, reported across different graph configurations.}
      \label{tab:ablation_optimizer}
      \begin{center}
        \begin{small}
          \begin{sc}
            \begin{tabular}{lcc}
              \toprule
              $(n, d)$ & MGA & GA \\
              \midrule
              $(1000, 500)$ & \textbf{130969.00} & 130560.33 \\
              $(1000, 800)$ & \textbf{204817.67} & 204576.33 \\
              $(2000, 500)$ & \textbf{264524.67} & 264351.00 \\
              \bottomrule
            \end{tabular}
          \end{sc}
        \end{small}
      \end{center}
      \vskip -0.1in
    \end{table}

Table~\ref{tab:ablation_optimizer} presents the results for the MaxCut problem using different graph as indicated in the first column. As observed, MGA consistently outperforms GA.

\section{Algorithms: mQO CPU and GPU Pseudocodes}\label{appen: implementation of pmQO}

Algorithm~\ref{alg: main mQO-Cut} presents the detailed procedure of mQO-MaxCut. Algorithm~\ref{alg: main pmQO-MIS} and Algorithm~\ref{alg: main pmQO-Cut} provide the detailed procedures for our GPU parallelized versions of mQO. 


$\mathcal{G}^{(b)}$ denotes the gradient direction evaluated at batch index $b$ of
$\mathbf{x}^{(b)}$, i.e., $\mathcal{G}^{(b)}=\nabla h(\mathbf{x}^{(b)})$
for MIS and $\mathcal{G}^{(b)}=\nabla f_{\textrm{B}}(\mathbf{x}^{(b)})$
for MaxCut. $\mathcal{P}_K$ denotes the pool of the top-$K$ retained solutions. $\operatorname{TopK}(\cdot)$ returns the best $K$ solutions ranked by cardinality (for MIS) or cut value (for MaxCut). $R^{(b)}$ denotes a randomly sampled subset of $V$ with size
$\lfloor \rho n \rfloor$.

\subsection{Details of the adopted discrete local search procedures}\label{sec: append detailed local search}

Here, we describe how we apply the local search procedures on MIS (resp. Maxcut) $I$ (resp. $S$), denoted in our algorithms as $\texttt{OneTwoSwap}(I,G)$ (resp. $\texttt{OneTwoFlip}(S,G)$).  

For MIS, we iterate over all $v\in V$ and check if the condition of $(1,2)$-swap repairability applies. If it does, we make the replacement and repeats iterating over all the nodes again. This is run until the $(1,2)$-swap rule no longer applies. This procedure is introduced in \cite{andrade2012fast}, and utilized in many later studies including the well-known heuristic, ReduMIS \cite{lamm2016finding}. 

For MaxCut, we apply an iterative local search that alternatively applies the $2$-flip procedure (from Definition~\ref{def: 2-flip}) that involves the $1$-flip procedure (from Definition~\ref{def: 1-flip}). While we have shown that PGA can escape $1$-flip repairable points when we use $f_\textrm{B}$, this $1$-flip is a part of the $2$-flip algorithm, applied outside the global reset loop in Step~11.

\begin{algorithm}[t]
\caption{\textbf{mQO-MaxCut}.}
\textbf{Input}: Graph $G$, initial $\mathbf{x}$, gradient operator $\mathcal{G}$ w.r.t.\ $f_{\textrm{B}}$ in \eqref{eqn: adj_lin_two_props}, reset parameter $\rho$, step size $\alpha$, number of global search rounds $T_{\mathrm{gs}}$, and $S_{\textrm{sol}}=\{\cdot\}$ (current best set). \\
\vspace{1mm}
\textbf{Output}: The best set solution $S_{\textrm{sol}}$\\
\vspace{1mm}
\small{01:} \textbf{While} time budget permits \\
\vspace{1mm}
\small{02:} \hspace{5mm}\textbf{Initialize} $\mathbf{x}$ as
$\mathbf{x} \leftarrow \Pi_{[-1,1]^n}(\mathbf{d}_{\text{base}} + \boldsymbol{\epsilon})$,
where $\boldsymbol{\epsilon} \sim \mathcal{N}(\mathbf{0}, \sigma^2 \mathbf{I})$ and
$\mathbf{d}_{\text{base},v} = 2(1 - \frac{\textrm{d}(v)}{\Delta}) - 1$ \\
\vspace{1mm}
\small{03:} \hspace{5mm}\textbf{Obtain} $\mathbf{x}\leftarrow \Pi_{[-1,1]^n}(\mathbf{x}+\alpha \mathcal{G})$ until convergence (optimizer gradient updates)\\
\vspace{1mm}
\small{04:} \hspace{5mm}\textbf{Obtain} $S = \{v\in V : \mathbf{x}_v > 0\}$ and \textbf{update} best set as $S_{\textrm{sol}} \leftarrow S$ only \textbf{If} $\texttt{Cut}(S_{\textrm{sol}}) < \texttt{Cut}(S)$ \\
\vspace{1mm}
\small{05:} \hspace{5mm}\textbf{For} $T_{\mathrm{gs}}$ rounds (global reset loop) \\
\vspace{1mm}
\small{06:} \hspace{12mm}\textbf{Obtain} $\mathbf{x}$ from $S_{\textrm{sol}}$ such that $\mathbf{x}_v = 1$ if $v \in S_{\textrm{sol}}$ and $\mathbf{x}_v =-1$ if $v \notin S_{\textrm{sol}}$\\
\vspace{1mm}
\small{07:} \hspace{12mm}\textbf{Obtain} set $R\subset V$ by randomly choosing $\lfloor\rho n\rfloor$ indices\\
\vspace{1mm}
\small{08:} \hspace{12mm}\textbf{Set} $\mathbf{x}_v \leftarrow 0,\ \forall v\in R$ (reset for global search)\\
\vspace{1mm}
\small{09:} \hspace{12mm}\textbf{Obtain} $\mathbf{x}\leftarrow \Pi_{[-1,1]^n}(\mathbf{x}+\alpha \mathcal{G})$ until convergence (optimizer gradient updates)\\
\vspace{1mm}
\small{10:} \hspace{12mm}\textbf{Obtain} $S' = \{v\in V : \mathbf{x}_v > 0\}$ and \textbf{update} best set as $S_{\textrm{sol}} \leftarrow S'$ only \textbf{If} $\texttt{Cut}(S_{\textrm{sol}}) < \texttt{Cut}(S')$\\
\vspace{1mm}
\small{11:} \hspace{5mm}\textbf{Obtain} $S \leftarrow \texttt{OneTwoFlip}(S_{\textrm{sol}},G)$ and \textbf{update} $S_{\textrm{sol}} \leftarrow S$ only \textbf{If} $\texttt{Cut}(S_{\textrm{sol}}) < \texttt{Cut}(S)$ (local search)\\
\vspace{1mm}
\vspace{-3.5mm}
\label{alg: main mQO-Cut}
\end{algorithm}

\begin{algorithm*}[h]
\caption{\textbf{parallelized mQO-MIS (pmQO-MIS, the GPU version of our MIS algorithm)}.}
\textbf{Input}: Graph $G$, initial $\mathbf{x}$, gradient operator $\mathcal{G}$ w.r.t. the MIS function $h$ in \eqref{eqn: MIS QUBO}, reset parameter $\rho$, step size $\alpha$, number of global search rounds $T_{\mathrm{gs}}$, batch size $B$, number of retained solutions $K$, and $I_{\textrm{sol}}=\{\cdot\}$ (current best set) \\
\vspace{1mm}
\textbf{Output}: The best obtained solution set $I_{\textrm{sol}}$\\
\vspace{1mm}
\small{01:} \textbf{While} time budget permits \\
\vspace{1mm}
\small{02:} \hspace{5mm}\textbf{Initialize} a batch of vectors $\{\mathbf{x}^{(b)}\}_{b=1}^{B}$ as
$\mathbf{x}^{(b)} \leftarrow \Pi_{[0,1]^n}(\mathbf{d}_{\text{base}}+\boldsymbol{\epsilon}^{(b)})$,
where $\mathbf{d}_{\text{base},v}=1 - \frac{\textrm{d}(v)}{\Delta}$ and
$\boldsymbol{\epsilon}^{(b)}\sim\mathcal{N}(\mathbf{0},\sigma^2 I)$ \\
\vspace{1mm}
\small{03:} \hspace{5mm}\textbf{Obtain} $\mathbf{x}^{(b)}\leftarrow \Pi_{[0,1]^n}(\mathbf{x}^{(b)}+\alpha \mathcal{G}^{(b)})$ in parallel for $b=1,\ldots,B$ until
$I^{(b)}=\{v\in V:\mathbf{x}^{(b)}_v>0\}$ is a MIS. \\
\vspace{1mm}
\small{04:} \hspace{5mm}\textbf{Obtain} $I^{(b)}=\{v\in V:\mathbf{x}^{(b)}_v=1\}$ for $b=1,\ldots,B$ and update the top-$K$ solution pool (where solutions are ranked by $|I^{(b)}|$)
\[
\mathcal{P}_K \leftarrow \operatorname{TopK}\left(\mathcal{P}_K \cup \{I^{(b)}\}_{b=1}^{B}\right),
\]
 \\
\vspace{1mm}
\small{05:} \hspace{5mm}\textbf{Update} best set as $I_{\textrm{sol}}\leftarrow I$ only \textbf{If} there exists $I\in\mathcal{P}_K$ such that $|I_{\textrm{sol}}|<|I|$ \\
\vspace{1mm}
\small{06:} \hspace{5mm}\textbf{For} $T_{\mathrm{gs}}$ rounds \textbf{do} \\
\vspace{1mm}
\small{07:} \hspace{12mm}\textbf{Sample} base solutions $\{\bar{I}^{(b)}\}_{b=1}^{B}$ from $\mathcal{P}_K$ and construct
$\bar{\mathbf{x}}^{(b)}_v=\mathds{1}(v\in \bar{I}^{(b)})$ \\
\vspace{1mm}
\small{08:} \hspace{12mm}\textbf{Obtain} $R^{(b)}\subset V$ by randomly choosing $\lfloor \rho n\rfloor$ indices for each $b$ \\
\vspace{1mm}
\small{09:} \hspace{12mm}\textbf{Set} $\bar{\mathbf{x}}^{(b)}_v\leftarrow 0,\ \forall v\in R^{(b)}$ for $b=1,\ldots,B$ \\
\vspace{1mm}
\small{10:} \hspace{12mm}\textbf{Obtain} $\mathbf{x}^{(b)}\leftarrow \Pi_{[0,1]^n}(\bar{\mathbf{x}}^{(b)}+\alpha \mathcal{G}^{(b)})$ in parallel until a MIS is detected \\
\vspace{1mm}
\small{11:} \hspace{12mm}\textbf{Obtain} $I'^{(b)}=\{v\in V:\mathbf{x}^{(b)}_v>0\}$ for $b=1,\ldots,B$ and update
\[
\mathcal{P}_K \leftarrow \operatorname{TopK}\left(\mathcal{P}_K \cup \{I'^{(b)}\}_{b=1}^{B}\right)
\]
\\
\vspace{1mm}
\small{12:} \hspace{12mm}\textbf{Update} best set as $I_{\textrm{sol}}\leftarrow I$ only \textbf{If} there exists $I\in\mathcal{P}_K$ such that $|I_{\textrm{sol}}|<|I|$ \\
\vspace{1mm}
\small{13:} \hspace{5mm}\textbf{Obtain} $I \leftarrow \texttt{OneTwoSwap}(I_{\textrm{sol}},G)$ and update best set as $I_{\textrm{sol}}\leftarrow I$ only \textbf{If} $|I_{\textrm{sol}}|<|I|$ \\
\vspace{1mm}
\vspace{-3.5mm}
\label{alg: main pmQO-MIS}
\end{algorithm*}

\begin{algorithm*}[h]
\caption{\textbf{parallelized mQO-MaxCut (pmQO-MaxCut, the GPU version of our MaxCut algorithm)}.}
\textbf{Input}: Graph $G$, initial $\mathbf{x}$, gradient operator $\mathcal{G}$ w.r.t.\ $f_{\textrm{B}}$ in \eqref{eqn: adj_lin_two_props}, reset parameter $\rho$, step size $\alpha$, number of global search rounds $T_{\mathrm{gs}}$, batch size $B$, number of retained solutions $K$, and $S_{\textrm{sol}}=\{\cdot\}$ (current best set). \\
\vspace{1mm}
\textbf{Output}: The best set solution $S_{\textrm{sol}}$\\
\vspace{1mm}
\small{01:} \textbf{While} time budget permits \\
\vspace{1mm}
\small{02:} \hspace{5mm}\textbf{Initialize} a batch of vectors $\{\mathbf{x}^{(b)}\}_{b=1}^{B}$ as
$\mathbf{x}^{(b)} \leftarrow \Pi_{[-1,1]^n}(\mathbf{d}_{\text{base}}+\boldsymbol{\epsilon}^{(b)})$,
where $\boldsymbol{\epsilon}^{(b)}\sim\mathcal{N}(\mathbf{0},\sigma^2\mathbf{I})$ and
$\mathbf{d}_{\text{base},v}=2(1-\frac{\textrm{d}(v)}{\Delta})-1$ \\
\vspace{1mm}
\small{03:} \hspace{5mm}\textbf{Obtain} $\mathbf{x}^{(b)}\leftarrow \Pi_{[-1,1]^n}(\mathbf{x}^{(b)}+\alpha \mathcal{G}^{(b)})$ in parallel for $b=1,\ldots,B$ until convergence \\
\vspace{1mm}
\small{04:} \hspace{5mm}\textbf{Obtain} $S^{(b)}=\{v\in V:\mathbf{x}^{(b)}_v>0\}$ for $b=1,\ldots,B$ and update the top-$K$ pool (solutions are ranked by $\texttt{Cut}(S^{(b)})$)
\[
\mathcal{P}_K \leftarrow \operatorname{TopK}\left(\mathcal{P}_K \cup \{S^{(b)}\}_{b=1}^{B}\right),
\]
 \\
\vspace{1mm}
\small{05:} \hspace{5mm}\textbf{Set}
$S_{\textrm{sol}} \leftarrow \arg\max_{S\in\mathcal{P}_K}\texttt{Cut}(S)$
only \textbf{If}
$\texttt{Cut}(S_{\textrm{sol}})<\max_{S\in\mathcal{P}_K}\texttt{Cut}(S)$ \\
\vspace{1mm}
\small{06:} \hspace{5mm}\textbf{For} $T_{\mathrm{gs}}$ rounds (global reset loop) \\
\vspace{1mm}
\small{07:} \hspace{10mm}\textbf{Sample} base solutions $\{\bar{S}^{(b)}\}_{b=1}^{B}$ from $\mathcal{P}_K$ and construct
$\bar{\mathbf{x}}^{(b)}_v = 1$ ($-1$) if $v \in$ ($\notin$) $\bar{S}^{(b)}$ \\
\vspace{1mm}
\small{08:} \hspace{10mm}\textbf{Obtain} $R^{(b)}\subset V$ by randomly choosing $\lfloor\rho n\rfloor$ indices for each $b$ \\
\vspace{1mm}
\small{09:} \hspace{10mm}\textbf{Set} $\bar{\mathbf{x}}^{(b)}_v\leftarrow 0,\ \forall v\in R^{(b)}$ for $b=1,\ldots,B$ \\
\vspace{1mm}
\small{10:} \hspace{10mm}\textbf{Obtain} $\mathbf{x}^{(b)}\leftarrow \Pi_{[-1,1]^n}(\bar{\mathbf{x}}^{(b)}+\alpha \mathcal{G}^{(b)})$ in parallel until convergence \\
\vspace{1mm}
\small{11:} \hspace{10mm}\textbf{Obtain} $S'^{(b)}=\{v\in V:\mathbf{x}^{(b)}_v>0\}$ for $b=1,\ldots,B$ and update 
\[
\mathcal{P}_K \leftarrow \operatorname{TopK}\left(\mathcal{P}_K \cup \{S'^{(b)}\}_{b=1}^{B}\right)
\]
 \\
\vspace{1mm}
\small{12:} \hspace{10mm}\textbf{Set}
$S_{\textrm{sol}} \leftarrow \arg\max_{S\in\mathcal{P}_K}\texttt{Cut}(S)$
only \textbf{If}
$\texttt{Cut}(S_{\textrm{sol}})<\max_{S\in\mathcal{P}_K}\texttt{Cut}(S)$ \\
\vspace{1mm}
\small{13:} \hspace{5mm}\textbf{For each} $S\in\mathcal{P}_K$, \textbf{obtain} $\widetilde{S}\leftarrow\texttt{OneTwoFlip}(S,G)$ and \textbf{update} best set as $S_{\textrm{sol}}\leftarrow \widetilde{S}$ only \textbf{If} $\texttt{Cut}(S_{\textrm{sol}})<\texttt{Cut}(\widetilde{S})$ \\
\vspace{1mm}
\vspace{-3.5mm}
\label{alg: main pmQO-Cut}
\end{algorithm*}





\section{Impact of $\lambda$ in $f_{\mathrm{B}}$}\label{appen: lambda ablation}
\begin{table}[htp]
\caption{Ablation study on the impact of the global reset ratio $\rho$ for MIS.
Each value corresponds to the averaged best result with 8 seeds and a run-time budget of 5 minutes.}
\vspace{-0.01cm}
\label{tab:ablation_global_mis}
\centering
\begin{small}
\begin{sc}
\begin{tabular}{lccc}
\toprule
$\rho$ & $(3000,100)$ & $(3000,300)$ & $(3000,1000)$ \\
\midrule
0.0 & 184.75 \tiny{$\pm$ 1.49} & 74.00 \tiny{$\pm$ 0.76} & 24.50 \tiny{$\pm$ 0.53} \\
0.1 & 181.88 \tiny{$\pm$ 2.70} & 75.12 \tiny{$\pm$ 1.81} & 24.50 \tiny{$\pm$ 0.53} \\
0.2 & 185.62 \tiny{$\pm$ 0.92} & 75.25 \tiny{$\pm$ 0.71} & 24.62 \tiny{$\pm$ 0.52} \\
0.3 & 187.63 \tiny{$\pm$ 1.60} & 77.12 \tiny{$\pm$ 0.83} & 25.25 \tiny{$\pm$ 0.46} \\
0.4 & 192.38 \tiny{$\pm$ 1.69} & 78.00 \tiny{$\pm$ 1.07} & 25.50 \tiny{$\pm$ 0.53} \\
0.5 & 195.50 \tiny{$\pm$ 2.00} & 76.25 \tiny{$\pm$ 0.93} & \textbf{26.00 \tiny{$\pm$ 0.00}} \\
0.6 & \textbf{197.50 \tiny{$\pm$ 0.93}} & \textbf{79.50 \tiny{$\pm$ 0.76}} & 25.88 \tiny{$\pm$ 0.35} \\
0.7 & 191.50 \tiny{$\pm$ 2.20} & 78.75 \tiny{$\pm$ 0.89} & 25.38 \tiny{$\pm$ 0.74} \\
0.8 & 188.00 \tiny{$\pm$ 1.07} & 77.12 \tiny{$\pm$ 1.13} & 25.75 \tiny{$\pm$ 0.46} \\
0.9 & 184.75 \tiny{$\pm$ 1.28} & 75.00 \tiny{$\pm$ 0.93} & 23.88 \tiny{$\pm$ 0.35} \\
\bottomrule
\end{tabular}
\end{sc}
\end{small}
\vskip -0.1in
\end{table}

\begin{table}[t]
\caption{Ablation study on the impact of the global reset ratio $\rho$ for the MaxCut problem.
Each value corresponds to the averaged best result with 3 seeds and a run-time budget of 10 minutes.}
\vspace{-0.0cm}
\label{tab:ablation_global_maxcut}
\centering
\begin{small}
\setlength{\tabcolsep}{3pt}
\begin{sc}
\begin{tabular}{lccc}
\toprule
$\rho$ & $(1000,500)$ & $(1000,800)$ & $(2000,1000)$ \\
\midrule
0.0 & 130660.33 \tiny{$\pm$ 112.19} & 204457.67 \tiny{$\pm$ 38.31} & 515972.67 \tiny{$\pm$ 188.26} \\
0.1 & 130622.00 \tiny{$\pm$ 87.18}  & 204550.00 \tiny{$\pm$ 23.04} & 515764.00 \tiny{$\pm$ 136.52} \\
0.2 & 130700.67 \tiny{$\pm$ 86.26}  & 204569.67 \tiny{$\pm$ 61.78} & 515660.67 \tiny{$\pm$ 232.65} \\
0.3 & 130783.00 \tiny{$\pm$ 93.19}  & 204610.33 \tiny{$\pm$ 77.55} & 515978.00 \tiny{$\pm$ 207.56} \\
0.4 & 130752.33 \tiny{$\pm$ 114.77} & 204617.00 \tiny{$\pm$ 82.96} & 516070.33 \tiny{$\pm$ 155.01} \\
0.5 & 130840.00 \tiny{$\pm$ 127.90} & 204667.67 \tiny{$\pm$ 76.85} & 516175.00 \tiny{$\pm$ 245.38} \\
0.6 & 130810.67 \tiny{$\pm$ 109.09} & 204712.00 \tiny{$\pm$ 108.84} & 516482.67 \tiny{$\pm$ 223.99} \\
0.7 & 130889.00 \tiny{$\pm$ 140.35} & 204742.67 \tiny{$\pm$ 104.08} & 516485.67 \tiny{$\pm$ 197.22} \\
0.8 & \textbf{130905.67 \tiny{$\pm$ 135.31}} & \textbf{204748.67 \tiny{$\pm$ 94.47}} & \textbf{516621.00 \tiny{$\pm$ 204.95}} \\
0.9 & 130860.00 \tiny{$\pm$ 63.76}  & 204696.00 \tiny{$\pm$ 78.30} & 516487.67 \tiny{$\pm$ 107.95} \\
\bottomrule
\end{tabular}
\end{sc}
\end{small}
\vskip -0.1in
\end{table}

\begin{figure}[h]
    \centering
    \includegraphics[width=0.4\linewidth]{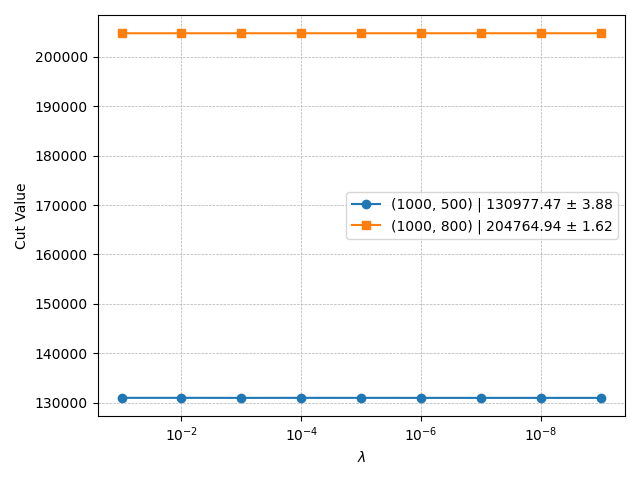}
    \caption{Impact of $\lambda$ on mQO-MaxCut using two ER graphs with $n=1000$ and two densities.}
    \label{fig: lambda}
\end{figure}

Here, we report the impact of $\lambda$ in the perturbed formulation, $f_\textrm{B}$, using two ER graphs with $n=1000$ and two densities. As shown in Figure~\ref{fig: lambda}, the performance of mQO-MaxCut remains stable across several orders of magnitude of $\lambda$. In both graph settings, the small standard deviation (given in the legend) shows that a wide range of $\lambda$ values returns very similar results. This indicates that the choice of $\lambda$ does not impact mQO's performance as long as it satisfies the condition in Theorem~\ref{thm: adj_lin_two_props}.


\section{Hyperparameter Settings}
\label{appen: hyperparameter settings}

Tables~\ref{tab:er_mis_hyperparameters}, ~\ref{tab:er_maxcut_hyperparameters}, \ref{tab:sbm_mis_para},\ref{tab:rb_mis_para}, \ref{tab:ba_maxcut_para} and \ref{tab:dimacs_mis_para} provide the hyperparameter selections of the graphs used in our experiments.

\begin{table}[h]
\centering
\caption{Hyperparameter settings for mQO-MIS and pmQO-MIS on ER graphs.}
\label{tab:er_mis_hyperparameters}
\begin{center}
    \begin{small}
        \begin{sc}
            \begin{tabular}{lcccc}
            \toprule
            $(n,d)$ & $\alpha$ & Momentum & $\rho$ & $T_{\mathrm{gs}}$ \\
            \midrule
            $(1000, 100)$     & 0.80 & 0.30 & 0.70 & $60$ \\
            $(1000, 300)$     & 0.80 & 0.45 & 0.70 & $60$ \\
            $(1000, 500)$     & 0.80 & 0.45 & 0.60 & $60$ \\
            $(3000, 100)$     & 0.80 & 0.30 & 0.60 & $60$ \\
            $(3000, 300)$     & 0.80 & 0.45 & 0.60 & $60$ \\
            $(3000, 1000)$    & 0.80 & 0.45 & 0.50 & $60$ \\
            $(10000, 5000)$   & 0.80 & 0.75 & 0.50 & $60$ \\
            $(20000, 10000)$  & 0.80 & 0.75 & 0.50 & $60$ \\
            $(30000, 15000)$  & 0.80 & 0.75 & 0.50 & $60$ \\
            \bottomrule
        \end{tabular}
      \end{sc}
    \end{small}
  \end{center}
  \vskip -0.1in
\end{table}

\begin{table}[h]
  \caption{Hyperparameter settings for mQO-MaxCut and pmQO-MaxCut on ER graphs.}
  \label{tab:er_maxcut_hyperparameters}
  \begin{center}
    \begin{small}
      \begin{tabular}{lcccc}
        \toprule
        $(n,d)$ & $\alpha$ & Momentum & $\rho$ & $T_{\mathrm{gs}}$  \\
        \midrule
        $(100, 50)$        & 0.0025 & 0.9 & 0.80 & $90$  \\
        $(1000, 100)$      & 0.0025 & 0.8 & 0.80 & $90$  \\
        $(1000, 500)$      & 0.0025 & 0.8 & 0.80 & $90$  \\
        $(1000, 800)$      & 0.0025 & 0.8 & 0.80 & $90$  \\
        $(30000, 15000)$   & $5\times10^{-5}$ & 0.8 & 0.80 & $90$  \\
        $(30000, 24000)$   & $5\times10^{-5}$ & 0.8 & 0.80 & $90$  \\
        $(40000, 20000)$   & $5\times10^{-5}$ & 0.8 & 0.80 & $90$  \\
        $(40000, 32000)$   & $5\times10^{-5}$ & 0.8 & 0.80 & $90$  \\
        \bottomrule
      \end{tabular}
    \end{small}
  \end{center}
  \vskip -0.1in
\end{table}

\begin{table}[h]
\centering
\caption{Hyperparameter settings for mQO-MIS on SBMs (2 clusters and $n=5000$).}
\label{tab:sbm_mis_para}
\resizebox{0.49\textwidth}{!}{
\begin{tabular}{lcccc}
\toprule
$(p_{in}, p_{out})$ & $\alpha$ & Momentum & $\rho$ & $T_{\mathrm{gs}}$ \\
\midrule
(0.10, 0.05) & 0.80 & 0.45 & 0.60 & 60 \\
(0.20, 0.05) & 0.80 & 0.45 & 0.60 & 60 \\
(0.30, 0.05) & 0.80 & 0.45 & 0.60 & 60 \\
(0.40, 0.05) & 0.80 & 0.45 & 0.60 & 60 \\
(0.50, 0.05) & 0.80 & 0.45 & 0.60 & 60 \\
\bottomrule
\end{tabular}}
\end{table}

\begin{table}[h]
\centering
\caption{Hyperparameter settings for mQO-MIS on large RB graphs.}
\label{tab:rb_mis_para}
\resizebox{0.49\textwidth}{!}{
\begin{tabular}{l|cccc}
\toprule
RB $(n = n'p', d)$ & $\alpha$ & Momentum & $\rho$ & $T_{\mathrm{gs}}$ \\
\midrule
(3000 = 30$\times$100, 1500) & 0.80 & 0.50 & 0.50 & 60  \\
(10000 = 100$\times$100, 5000) & 0.80 & 0.70 & 0.50 & 60 \\
\bottomrule
\end{tabular}}
\end{table}

\begin{table}[h]
\centering
\caption{Hyperparameter settings for MaxCut on large BA graphs.}
\label{tab:ba_maxcut_para}
\resizebox{0.49\textwidth}{!}{
\begin{tabular}{lcccc}
\toprule
BA $(n, m')$ & $\alpha$ & Momentum & $\rho$ & $T_{\mathrm{gs}}$ \\
\midrule
(3000, 1500) & 0.001 & 0.80 & 0.80 & 90 \\
(30000, 15000) & $5\times10^{-5}$ & 0.80 & 0.80 & 90 \\
\bottomrule
\end{tabular}}
\end{table}

\begin{table*}[h]
\centering
\caption{Hyperparameter settings for mQO-MIS on DIMACS benchmark graphs.}
\label{tab:dimacs_mis_para}
\resizebox{0.4\textwidth}{!}{
\begin{tabular}{lcccc}
\toprule
Graph Name & $\alpha$ & Momentum & $\rho$ & $T_{\mathrm{gs}}$ \\
\midrule
brock200-1 & 0.75 & 0.80 & 0.6 & 60 \\
brock200-2 & 0.75 & 0.80 & 0.6 & 60 \\
brock200-3 & 0.75 & 0.80 & 0.6 & 60 \\
brock200-4 & 0.75 & 0.80 & 0.6 & 60 \\
brock400-1 & 0.75 & 0.80 & 0.6 & 60 \\
brock400-2 & 0.75 & 0.80 & 0.6 & 60 \\
brock400-3 & 0.75 & 0.80 & 0.6 & 60 \\
brock400-4 & 0.75 & 0.80 & 0.6 & 60 \\
brock800-1 & 0.75 & 0.80 & 0.6 & 60 \\
brock800-2 & 0.75 & 0.80 & 0.6 & 60 \\
brock800-3 & 0.75 & 0.80 & 0.6 & 60 \\
brock800-4 & 0.75 & 0.80 & 0.6 & 60 \\
C125-9 & 0.75 & 0.80 & 0.6 & 60 \\
C250-9 & 0.75 & 0.80 & 0.6 & 60 \\
C500-9 & 0.75 & 0.80 & 0.6 & 60 \\
C1000-9 & 0.75 & 0.80 & 0.6 & 60 \\
C2000-5 & 0.75 & 0.80 & 0.6 & 60 \\
C2000-9 & 0.75 & 0.80 & 0.6 & 60 \\
C4000-5 & 0.75 & 0.80 & 0.6 & 60 \\
c-fat200-1 & 0.001 & 0.1 & 0.6 & 60 \\
c-fat200-2 & 0.001 & 0.1 & 0.6 & 60 \\
c-fat200-5 & 0.001 & 0.1 & 0.6 & 60 \\
c-fat500-1 & 0.001 & 0.1 & 0.6 & 60 \\
c-fat500-2 & 0.001 & 0.1 & 0.6 & 60 \\
c-fat500-5 & 0.001 & 0.1 & 0.6 & 60 \\
c-fat500-10 & 0.001 & 0.1 & 0.6 & 60 \\
DSJC500-5 & 0.75 & 0.80 & 0.6 & 60 \\
DSJC1000-5 & 0.75 & 0.80 & 0.6 & 60 \\
gen200-p0-9-44 & 0.75 & 0.80 & 0.6 & 60 \\
gen200-p0-9-55 & 0.75 & 0.80 & 0.6 & 60 \\
gen400-p0-9-55 & 0.75 & 0.80 & 0.6 & 60 \\
gen400-p0-9-65 & 0.75 & 0.80 & 0.6 & 60 \\
gen400-p0-9-75 & 0.75 & 0.80 & 0.6 & 60 \\
hamming6-2 & 0.75 & 0.80 & 0.6 & 60 \\
hamming6-4 & 0.75 & 0.80 & 0.6 & 60 \\
hamming8-2 & 0.75 & 0.80 & 0.6 & 60 \\
hamming8-4 & 0.75 & 0.80 & 0.6 & 60 \\
hamming10-2 & 0.001 & 0.1 & 0.6 & 60 \\
hamming10-4 & 0.75 & 0.80 & 0.6 & 60 \\
johnson8-2-4 & 0.75 & 0.80 & 0.6 & 60 \\
johnson8-4-4 & 0.75 & 0.80 & 0.6 & 60 \\
johnson16-2-4 & 0.75 & 0.80 & 0.6 & 60 \\
johnson32-2-4 & 0.75 & 0.80 & 0.6 & 60 \\
keller4 & 0.75 & 0.80 & 0.6 & 60 \\
keller5 & 0.75 & 0.80 & 0.6 & 60 \\
keller6 & 0.75 & 0.80 & 0.6 & 60 \\
MANN-a9 & 0.75 & 0.80 & 0.6 & 60 \\
MANN-a27 & 0.75 & 0.80 & 0.6 & 60 \\
MANN-a45 & 0.001 & 0.1 & 0.6 & 60 \\
MANN-a81 & 0.001 & 0.1 & 0.6 & 60 \\
p-hat300-3 & 0.75 & 0.80 & 0.6 & 60 \\
p-hat500-2 & 0.75 & 0.80 & 0.6 & 60 \\
p-hat500-3 & 0.75 & 0.80 & 0.6 & 60 \\
p-hat700-1 & 0.75 & 0.80 & 0.6 & 60 \\
p-hat700-2 & 0.75 & 0.80 & 0.6 & 60 \\
p-hat700-3 & 0.75 & 0.80 & 0.6 & 60 \\
p-hat1000-1 & 0.75 & 0.80 & 0.6 & 60 \\
p-hat1000-2 & 0.75 & 0.80 & 0.6 & 60 \\
p-hat1500-1 & 0.75 & 0.80 & 0.6 & 60 \\
p-hat1500-2 & 0.75 & 0.80 & 0.6 & 60 \\
p-hat1500-3 & 0.75 & 0.80 & 0.6 & 60 \\
\bottomrule
\end{tabular}}
\end{table*}


\section{Limitations \& Future Work}\label{sec: append limiations and future work}

While mQO demonstrates strong performance on large-scale graphs under constrained run-time budgets, its advantages are less pronounced on smaller graph instances. In these regimes (e.g., smaller DIMACS as well as the first rows of Tables~\ref{tab:er_comparison mis} and~\ref{tab:er_comparison cut}), classical heuristics and ILP solvers can be more efficient and competitive. In addition, mQO relies on a small set of hyperparameters (e.g., step size and momentum), which may require tuning across problem instances. This sensitivity can introduce additional overhead compared to off-the-shelf ILP solvers. 

\end{document}